\def\cl@chapter{}

\documentclass{acmsmall}
\usepackage{graphicx}
\usepackage{placeins}
\usepackage{xspace}
\usepackage{xtab}
\usepackage[algo2e,ruled,vlined,linesnumbered,resetcount]{algorithm2e}

\SetAlFnt{\small}
\SetAlCapFnt{\small}
\SetAlCapNameFnt{\small}
\SetAlCapHSkip{0pt}
\IncMargin{-\parindent}

\usepackage{multicol}

\usepackage[T1]{fontenc}
\usepackage{pgfplotstable}
\usepackage{color}
\usepackage{etex}
\usepackage[english]{babel}
\usepackage{amssymb}
\usepackage{bibentry}
\usepackage{enumitem}
\usepackage{xytree}
\usepackage{amsmath}

\usetikzlibrary{arrows}

\usepackage{color,soul}

\newcommand{\ie}{i.\,e.\xspace}

\newcommand{\degree}{\operatorname{deg}}
\newcommand{\diam}{\operatorname{diam}}

\renewcommand{\O}{\mathcal O}
\newcommand{\preproc}{PreProcessing}

\newcommand{\deltabfs}{{\sc Olh}}
\newcommand{\sampl}{{\sc Ocl}}
\newcommand{\degcut}{{\sc DegCut}}
\newcommand{\degbound}{{\sc DegBound}}
\newcommand{\nbcut}{{\sc NBCut}}
\newcommand{\nbbound}{{\sc NBBound}}
\newcommand{\mvis}{m_{\text{vis}}}

\newcommand{\degr}[1]{\mathrm{deg}(#1)}
\newcommand{\odegr}[1]{\mathrm{outdeg}(#1)}

\newcommand{\Reac}[1]{R(#1)}
\newcommand{\reac}[1]{r(#1)}

\newcommand{\neig}[2]{\Gamma_{#1}(#2)}
\newcommand{\sneig}[2]{\gamma_{#1}(#2)}
\newcommand{\tsneig}[2]{\tilde{\gamma}_{#1}(#2)}
\newcommand{\uneig}[2]{\tilde{\gamma}_{#1}(#2)}

\newcommand{\ball}[2]{N_{#1}(#2)}
\newcommand{\sball}[2]{n_{#1}(#2)}

\newcommand{\farn}[1]{f(#1)}
\newcommand{\lfarns}{L}
\newcommand{\lfarn}[2]{\lfarns(#1,#2)}
\newcommand{\lfarncut}[3]{\lfarns^{\text{CUT}}_{#1}(#2,#3)}
\newcommand{\lfarnnb}[2]{\lfarns^{\text{NB}}(#1,#2)}
\newcommand{\lfarnnbr}[1]{\lfarns^{\text{NB}}(#1)}
\newcommand{\lfarnlb}[3]{\lfarns^{\text{LB}}_{#1}(#2,#3)}

\newcommand{\lsum}[2]{S(#1,#2)}
\newcommand{\lsumcut}[3]{S^{\text{CUT}}_{#1}(#2,#3)}
\newcommand{\lsumnb}[2]{S^{\text{NB}}(#1,#2)}
\newcommand{\lsumlb}[3]{S^{\text{LB}}_{#1}(#2,#3)}

\newcommand{\clos}[1]{c(#1)}

\newcommand{\cb}{\texttt{\upshape computeBounds}}
\newcommand{\ub}{\texttt{\upshape updateBounds}}
\newcommand{\cbnb}{\texttt{\upshape \cb NB}}
\newcommand{\cbdeg}{\texttt{\upshape \cb Deg}}
\newcommand{\ubcut}{\texttt{\upshape \ub BFSCut}}
\newcommand{\ublb}{\texttt{\upshape \ub LB}}

\newcommand{\weight}{w}
\DeclareMathOperator{\ecc}{ecc}

\DeclareMathOperator{\outdeg}{outdeg}

\newcommand{\C}{\mathcal{C}}

\begin{document}

% Page heads
\markboth{E. Bergamini et al.}{Computing top-$k$ Closeness Centrality Faster in Unweighted Graphs}

% Title portion
\title{Computing top-$k$ Closeness Centrality Faster in Unweighted Graphs}
\author{ELISABETTA BERGAMINI
\affil{Karlsruhe Institute of Technology (KIT)}
MICHELE BORASSI
\affil{IMT Institute for Advanced Studies Lucca}
PIERLUIGI CRESCENZI
\affil{Universit\`a di Firenze}
ANDREA MARINO
\affil{Universit\`a di Pisa}
HENNING MEYERHENKE
\affil{Karlsruhe Institute of Technology (KIT), Germany}}
% NOTE! Affiliations placed here should be for the institution where the
%       BULK of the research was done. If the author has gone to a new
%       institution, before publication, the (above) affiliation should NOT be changed.
%       The authors 'current' address may be given in the "Author's addresses:" block (below).
%       So for example, Mr. Abdelzaher, the bulk of the research was done at UIUC, and he is
%       currently affiliated with NASA.

\begin{abstract} 
%\small\baselineskip=9pt 
Given a connected graph $G=(V,E)$, the closeness centrality of a vertex $v$ is defined as $\frac{n-1}{\sum_{w \in V} d(v,w)}$. This measure is widely used in the analysis of real-world complex networks, and the problem of selecting the $k$ most central vertices has been deeply analysed in the last decade. However, this problem is computationally not easy, especially for large networks: in the first part of the paper, we prove that it is not solvable in time $\O(|E|^{2-\epsilon})$ on directed graphs, for any constant $\epsilon>0$, under reasonable complexity assumptions. Furthermore, we propose a new algorithm for selecting the $k$ most central nodes in a graph: we experimentally show that this algorithm improves significantly both the textbook algorithm, which is based on computing the distance between all pairs of vertices, and the state of the art. For example, we are able to compute the top $k$ nodes in few dozens of seconds in real-world networks with millions of nodes and edges. Finally, as a case study, we compute the $10$ most central actors in the IMDB collaboration network, where two actors are linked if they played together in a movie, and in the Wikipedia citation network, which contains a directed edge from a page $p$ to a page $q$ if $p$ contains a link to $q$. 
\end{abstract}


%\ccsdesc[300]{Theory of computation~Shortest paths}


% We no longer use \terms command
%\terms{Design, Algorithms, Performance}


% At a minimum you need to supply the author names, year and a title.
% IMPORTANT:
% Full first names whenever they are known, surname last, followed by a period.
% In the case of two authors, 'and' is placed between them.
% In the case of three or more authors, the serial comma is used, that is, all author names
% except the last one but including the penultimate author's name are followed by a comma,
% and then 'and' is placed before the final author's name.
% If only first and middle initials are known, then each initial
% is followed by a period and they are separated by a space.
% The remaining information (journal title, volume, article number, date, etc.) is 'auto-generated'.

\maketitle
%\author[1]{Michele Borassi}
%\author[2]{Pierluigi Crescenzi}
%\author[3]{Andrea Marino}
%\affil[1]{IMT Institute for Advanced Studies Lucca, Italy}
%\affil[2]{Dipartimento di Ingegneria dell'Informazione, Universit\`a di Firenze, Italy}
%\affil[3]{Dipartimento di Informatica, Universit\`a di Pisa, Italy}

%\category{G.2.2}{Discrete Mathematics}{Graph Theory}[graph algorithms]
%\category{H.2.8}{Database Management}{Database Applications}[data mining]
%\category{I.1.2}{Computing Methodologies}{Algorithms}[analysis of algorithms]
%\terms{Design, Algorithms, Performance}

%\keywords{Network Science, Centrality, Closeness, Complex Networks}

\pagestyle{plain}
\pagenumbering{arabic}
\setcounter{page}{0}

\newpage

\section{Introduction}
The problem of identifying the most central nodes in a network is a fundamental question that has been asked many times in a plethora of research areas, such as biology, computer science, sociology, and psychology. Because of the importance of this question, dozens of centrality measures have been introduced in the literature (for a recent survey, see~\cite{Boldi2014}). Among these measures, closeness centrality is certainly one of the oldest and of the most widely used~\cite{Bavelas1950}: almost all books dealing with network analysis discuss it (for example, \cite{Newman2010}), and almost all existing network analysis libraries implement algorithms to compute it. 

In a connected graph, the closeness centrality of a node $v$ is defined as $\clos v = \frac{n-1}{\sum_{w \in V} d(v,w)}$. The idea behind this definition is that a central node should be very efficient in spreading information to all other nodes: for this reason, a node is central if the average number of links needed to reach another node is small. If the graph is not (strongly) connected, researchers have proposed various ways to extend this definition: for the sake of simplicity, we focus on Lin's index, because it coincides with closeness centrality in the connected case and because it is quite established in the literature \cite{Lin1976,Wasserman1994,Boldi2013,Boldi2014,Olsen2014}. However, our algorithms can be adapted very easily to compute other possible generalizations, such as harmonic centrality \cite{Marchiori2000} and exponential centrality \cite{Wang2014} (see Sect. \ref{sec:defs} for more details).

In order to compute the $k$ vertices with largest closeness, the textbook algorithm computes $\clos v$ for each $v$ and returns the $k$ largest found values. The main bottleneck of this approach is the computation of $d(v,w)$ for each pair of vertices $v$ and $w$ (that is, solving the All Pairs Shortest Paths or APSP problem). This can be done in two ways: either by using fast matrix multiplication, in time $\O(n^{2.373}\log n)$ \cite{Zwick2002,Williams2012}, or by performing a breadth-first search (in short, BFS) from each vertex $v \in V$, in time $\O(mn)$, where $n=|V|$ and $m=|E|$. Usually, the BFS approach is preferred because the other approach contains big constants hidden in the $\O$ notation, and because real-world networks are usually sparse, that is, $m$ is not much bigger than $n$. However, also this approach is too time-consuming if the input graph is very big (with millions of nodes and hundreds of millions of edges). 

Our first result proves that, in the worst case, the BFS-based approach cannot be improved, under reasonable complexity assumptions. Indeed, we construct a reduction from the problem of computing the most central vertex (the case $k=1$) to the Orthogonal Vector problem \cite{Abboud2016}. This reduction implies that we cannot compute the most central vertex in $\O(m^{2-\epsilon})$ for any $\epsilon>0$, unless the Orthogonal Vector conjecture \cite{Abboud2016} is false. Note that the Orthogonal Vector conjecture is implied by the well-known Strong Exponential Time Hypothesis (SETH, \cite{Impagliazzo2001}), and hence all our results hold also if we assume SETH. This hypothesis is heavily used in the context of polynomial-time reductions, and, informally, it says that the \textsc{Satisfiability} problem is not solvable in time $\O((2-\epsilon)^N)$ for any $\epsilon>0$, where $N$ is the number of variables. 
Our result still holds if we assume the input graph to be sparse, that is, if we assume that $m=\O(n)$ (the general non-sparse case follows immediately;
of course, if the input graph is not sparse, then the BFS-based approach can be improved using fast matrix multiplication). The proof is provided in Sect.~\ref{sec:comp}.

Knowing that the BFS-based algorithm cannot be improved in the worst case, in the second part of the paper we provide a new exact algorithm that performs much better on real-world networks, making it possible to compute the $k$ most central vertices in networks with millions of nodes and hundreds of millions of edges. The new approach combines the BFS-based algorithm with a pruning technique: during the algorithm, we compute and update upper bounds on the closeness of all the nodes, and we exclude a node $v$ from the computation as soon as its upper bound is ``small enough'', that is, we are sure that $v$ does not belong to the top $k$ nodes. We propose two different strategies to set the initial bounds, and two different strategies to update the bounds during the computation: this means that our algorithm comes in four different variations. The experimental results show that different variations perform well on different kinds of networks, and the best variation of our algorithm drastically outperforms both a probabilistic approach \cite{Okamoto2008}, and the best exact algorithm available until now \cite{Olsen2014}. We have computed for the first time the $10$ most central nodes in networks with millions of nodes and hundreds of millions of edges, 
and do so in very little time. 
A significant example is the \texttt{wiki-Talk} network, which was also used in \cite{Saryuce2013}, where the authors propose an algorithm to update closeness centralities after edge additions or deletions. Our performance is about $30\,000$ times better than the performance of the textbook algorithm: if only the most central node is needed, we can recompute it from scratch more than $150$ times faster than the geometric average update time in \cite{Saryuce2013}. Moreover, our approach is not only very efficient, but it is also very easy to code, making it a very good candidate to be implemented in existing graph libraries. We provide an implementation of it in NetworKit \cite{Staudt2014} and of one of its variations in Sagemath \cite{Csardi2006}. We sketch the main ideas of the algorithm in Sect.~\ref{sec:overview}, and we provide all details in Sect.~\ref{sec:nblb}-\ref{sec:disc}. We experimentally evaluate the efficiency of the new algorithm in Sect.~\ref{sec:exp}. 

Also, our approach can be easily extended to any centrality measure in the form $c(v) = \sum_{w \neq v} f(d(v, w))$, where $f$ is a decreasing function. Apart from Lin's index, almost all the approaches that try to generalize closeness centrality to disconnected graphs fall under this category. The most popular among these measures is \textit{harmonic centrality} \cite{Marchiori2000}, defined as $h(v) = \sum_{w \neq v} \frac{1}{d(v, w)}$. For the sake of completeness, in Sect.~\ref{sec:exp} we show that our algorithm performs well also for this measure.

In the last part of the paper (Sect.~\ref{sec:imdb}, \ref{sec:wiki}), we consider two case studies: the actor collaboration network ($1\,797\,446$ vertices, $72\,880\,156$ edges) and the Wikipedia citation network ($4\,229\,697$ vertices, $102\,165\,832$ edges). In the actor collaboration network, we analyze the evolution of the $10$ most central vertices, considering snapshots taken every 5 years between 1940 and 2014. The computation was performed in little more than $45$ minutes. In the Wikipedia case study, we consider both the standard citation network, that contains a directed edge $(p,q)$ if $p$ contains a link to $q$, and the reversed network, that contains a directed edge $(p,q)$ if $q$ contains a link to $p$. For most of these graphs, we are able to compute the $10$ most central pages in a few minutes, making them available for further analyses.

\subsection{Related Work}

Closeness is a ``traditional'' definition of centrality, and consequently it was not ``designed with scalability in mind'', as stated in \cite{Kang2011}. Also in \cite{Chen2012}, it is said that closeness centrality can ``identify influential nodes'', but it is ``incapable to be applied in large-scale networks due to the computational complexity''. The simplest solution considered was to define different measures that might be related to closeness centrality \cite{Kang2011}.

\paragraph{\textbf{Hardness results}} A different line of research has tried to develop more efficient algorithms, or lower bounds for the complexity of this problem. In particular, in \cite{Borassi2015} it is proved that finding the least closeness-central vertex is not subquadratic-time solvable, unless SETH is false. In the same line, it is proved in \cite{Abboud2016} that finding the most central vertex is not solvable in $\O(m^{2-\epsilon})$, assuming the Hitting Set conjecture. This conjecture is very recent, and there are not strong evidences that it holds, apart from its similarity to the Orthogonal Vector conjecture. Conversely, the Orthogonal Vector conjecture is more established: it is implied both by the Hitting Set conjecture \cite{Abboud2016}, and by SETH \cite{Williams2005}, a widely used assumption in the context of polynomial-time reductions \cite{Impagliazzo2001,Williams2005,Williams2010,Patrascu2010,Roditty2013,Abboud2014,Abboud2014a,Abboud2015,Borassi2015,Abboud2016,Borassi2016}. Similar hardness results were also proved in the dense weighted context \cite{Abboud2015}, by linking the complexity of centrality measures to the complexity of computing the All Pairs Shortest Paths.

\paragraph{\textbf{Approximation algorithms}} In order to deal with the above hardness results, it is possible to design approximation algorithms: the simplest approach samples the distance between a node $v$ and $l$ other nodes $w$, and returns the average of all values $d(v,w)$ found \cite{Eppstein2004}. The time complexity is $\O(lm)$, to obtain an approximation $\tilde c(v)$ of the centrality of each node $v$ such that $\mathbb P\left(\left|\frac{1}{\tilde c(v)}-\frac{1}{\clos{v}}\right| \geq \epsilon D \right) \leq 2e^{-\Omega\left(l\epsilon^2\right)}$, where $D$ is the diameter of the graph (the diameter is the maximum distance between any two connected nodes). A more refined approximation algorithm is provided in \cite{Cohen2014}, which combines the sampling approach with a $3$-approximation algorithm: this algorithm has running time $\O(lm)$, and it provides an estimate $\tilde c(v)$ of the centrality of each node $v$ such that $\mathbb P\left(\left|\frac{1}{\tilde c(v)}-\frac{1}{\clos{v}}\right| \geq \frac{\epsilon}{\clos{v}} \right) \leq 2e^{-\Omega\left(l\epsilon^3\right)}$ (note that, differently from the previous algorithm, this algorithm provides a guarantee on the relative error). The most recent result by Chechik et al.~\cite{DBLP:conf/approx/ChechikCK15} allows to approximate closeness centrality with a coefficient of variation of $\epsilon$ using $\O(\epsilon^{-2})$ single-source shortest path (SSSP) computations. Alternatively, one can make the probability that the maximum relative error exceeds $\epsilon$ polynomially small by using $\O(\epsilon^{-2} \log n)$ SSSP computations.

%%% Michele's comment
%   Indeed, the problem is that, in real-world graphs, the farness of a node is between, say, 1 and 10, meaning that we need an error eps=10/n to obtain the ranking. As a consequence, if we want to use their algorithm for ranking, we need approximately time m/eps^2=mn^2/100, which is clearly not practical.

%However, even if these approximation algorithms work quite well, they are not suited to the ranking of nodes with closeness centrality: indeed, we work with so-called \emph{small-world} networks, having a low diameter. Consequently, in a typical graph, the average distance between $v$ and a random node $w$ is between 1 and 10, meaning that most of the $n$ centrality (actually: farness) values lie in this range. In order to obtain a ranking, we need the error to be close to $\frac{10}{n}$, which might be very small. As an example, performing $O(\epsilon^{-2})$ SSSPs as in~\cite{DBLP:conf/approx/ChechikCK15} would then require $O(m / \epsilon^2) = O(m n^2)$ time in the unweighted case, which is impractical for large graphs.

However, these approximation algorithms have not been specifically designed for ranking nodes according to their closeness centrality, and turning them into a trustable top-$k$ algorithm can be a challenging problem. Indeed, observe that, in many real-world cases, we work with so-called \emph{small-world} networks, having a low diameter. Hence, in a typical graph, the average distance between $v$ and a random node $w$ is between 1 and 10. This implies that most of the $n$ values $\frac{1}{\clos{v}}$ lie in this range, and that, in order to obtain a reliable ranking, we need the error to be close to $\epsilon=\frac{10}{n}$, which might be very small in the case of the vast majority of real-word networks. As an example, performing $\O(\epsilon^{-2})$ SSSPs as in~\cite{DBLP:conf/approx/ChechikCK15} would then require $\O(\frac{m}{\epsilon^2}) = \O(m n^2)$ time in the unweighted case, which is impractical for large graphs. In the absence of theoretical results, it is, however, worth noting that, as a side effect of our new algorithm, we can now quickly certify in practice, even in the case of very large graphs, how good the ranking produced by these approximation algorithms is. For example, if we run this algorithm on our dataset with the same number of iterations as our algorithm, the relative error guaranteed on the centrality of all the nodes is large (usually, above $50\%$ for $k=100$), because the algorithm is not tailored to the top-$k$ computation. However, with our algorithm one can show that the ranking obtained is very close to the correct one (usually, more than $95$ of the $100$ most central nodes according to~\cite{DBLP:conf/approx/ChechikCK15} are actually in the top-$100$).\footnote{Indeed, we obtained a similar experimental result while dealing with the simpler heuristics consisting in choosing as the sample a set of highest degree nodes slightly larger than the sample chosen by the algorithm in~\cite{DBLP:conf/approx/ChechikCK15}.} A theoretical justification of this behavior is, in our opinion, a very interesting open problem.
%Even if these results are not reported in this paper, we have, for instance, experimentally verified that the ranking produced by the algorithm in~\cite{DBLP:conf/approx/ChechikCK15} is almost always very close to the real ranking, and that at least 90\% of the top-$k$ nodes are correctly reported.

Finally, an approximation algorithm was proposed in \cite{Okamoto2008}, where the sampling technique developed in \cite{Eppstein2004} was used to actually compute the top $k$ vertices: the result is not exact, but it is exact with high probability. The authors proved that the time complexity of their algorithm is $\O(mn^{\frac{2}{3}}\log n)$, under the rather strong assumption that closeness centralities are uniformly distributed between $0$ and the diameter $D$ (in the worst case, the time complexity of this algorithm is $\O(mn)$).

\paragraph{\textbf{Heuristics}} Other approaches have tried to develop incremental algorithms that might be more suited to real-world networks. For instance, in \cite{Lim2011}, the authors develop heuristics to determine the $k$ most central vertices in a varying environment. Furthermore, in \cite{Saryuce2013}, the authors consider the problem of updating the closeness centrality of all nodes after edge insertions or deletions: in some cases, the time needed for the update could be orders of magnitude smaller than the time needed to recompute all centralities from scratch.

Finally, some works have tried to exploit properties of real-world networks in order to find more efficient algorithms. In \cite{LeMerrer2014}, the authors develop a heuristic to compute the $k$ most central vertices according to different measures. The basic idea is to identify central nodes according to a simple centrality measure (for instance, degree of nodes), and then to inspect a small set of central nodes according to this measure, hoping it contains the top $k$ vertices according to the ``complex'' measure. The last approach \cite{Olsen2014}, proposed by Olsen et al., tries to exploit the properties of real-world networks in order to develop exact algorithms with worst case complexity $\O(mn)$, but performing much better in practice. As far as we know, this is the only exact algorithm that is able to efficiently compute the $k$ most central vertices in networks with up to $1$ million nodes, before this work.

\paragraph{\textbf{Software libraries}} Despite this huge amount of research, graph libraries still use the textbook algorithm: among them, Boost Graph Library \cite{Hagberg2008}, igraph \cite{Stein2005} and NetworkX \cite{Siek2001}. This is due to the fact that efficient available exact algorithms for top-$k$ closeness centrality, like \cite{Olsen2014}, are relatively recent and make use of several other non-trivial routines.
We provide an implementation of the algorithm presented in this paper for
Sagemath \cite{Csardi2006} and NetworKit \cite{Staudt2014}.

\section{Preliminaries}\label{sec:defs}
\begin{table}[!t]
\begin{center}
\begin{footnotesize}
\tbl{Notations used throughout the paper.\label{tbl:notation}}{
\begin{tabular}{|p{2cm}|p{8cm}|}
\hline
\textbf{Symbol} & \textbf{Definition}\\
\hline\hline
\multicolumn{2}{|c|}{\textbf{Graphs}}\\
\hline
 $G=(V,E)$ & Graph with node/vertex set $V$ and edge/arc set $E$\\
\hline
$n$, $m$ & $|V|$, $|E|$\\
\hline
$\mathcal{G}=(\mathcal{V},\mathcal{E},\weight)$ & Weighted directed acyclic graph of strongly connected components (see Sect.~\ref{sec:alphaomega})\\
%\hline
%\hline
%\multicolumn{2}{|c|}{\textbf{Degree functions}}\\
\hline
$\degr v$ & Degree of a node in an undirected graph\\
\hline
$\odegr v$ & Out-degree of a node in a directed graph\\
%\hline
%\hline
%\multicolumn{2}{|c|}{\textbf{Distance function}}\\ 
\hline
$d(v,w)$ & Number of edges in a shortest path from $v$ to $w$\\
\hline
\hline
\multicolumn{2}{|c|}{\textbf{Reachability set function}}\\ 
\hline
$\Reac v$ & Set of nodes reachable from $v$ (by definition, $v\in \Reac v$)\\
\hline
$\reac v$ & $|\Reac v|$\\
\hline
$\alpha(v)$ & Lower bound on $\reac v$, that is, $\alpha(v) \leq \reac v$ (see Sect.~\ref{sec:alphaomega})\\
\hline
$\omega(v)$ & Upper bound on $\reac v$, that is, $\reac v \leq \omega(v)$ (see Sect.~\ref{sec:alphaomega}) \\
\hline
\hline
\multicolumn{2}{|c|}{\textbf{Neighborhood functions}}\\
\hline
$\neig dv$ & Set of nodes at distance $d$ from $v$: $\{w \in V: d(v,w)=d\}$\\
\hline
$\Gamma(v)$ & Set of neighbors of $v$, that is $\neig 1v$\\
\hline
$\sneig dv$ & Number of nodes at distance $d$ from $v$, that is, $|\neig dv|$\\
\hline
$\tsneig dv$ & Upper bound on $\sneig {d}v$ computed using the neighborhood-based lower bound (see Sect.~\ref{sec:nblb})\\
\hline
$\uneig {d+1}v$ & Upper bound on $\sneig {d+1}v$, defined as $\sum_{u \in \neig dv} \degr u-1$ if the graph is undirected, $\sum_{u \in \neig dv}\odegr u$ otherwise\\
%\hline
%\hline
%\multicolumn{2}{|c|}{\textbf{Ball functions}}\\  
\hline
$\ball dv$ & Set of nodes at distance \textit{at most} $d$ from $v$, that is, $\{w \in V: d(v ,w)\leq d\}$\\
\hline
$\sball dv$ & Number of nodes at distance \textit{at most} $d$ from $v$, that is, $|\ball dv|$\\
\hline
\hline
\multicolumn{2}{|c|}{\textbf{Closeness functions}}\\ 
\hline
$\clos v$ & Closeness of node $v$, that is, $\frac{(\reac v - 1)^2}{(n-1)\sum_{w \in \Reac v} d(v,w)}$\\
\hline
\hline
\multicolumn{2}{|c|}{\textbf{Distance sum functions}}\\ 
\hline
$S(v)$ & Total distance of node $v$, that is $\sum_{w \in \Reac v} d(v,w)$\\
\hline
$\lsumnb vr$ & Lower bound on $S(v)$ if $\reac v=r$, used in the \cbnb\ function (see Prop.~\ref{prop:lbound})\\
\hline
$\lsumcut dvr$ & Lower bound on $S(v)$ if $\reac v=r$, used in the \ubcut\ function (see Lemma~\ref{lem:boundokconn})\\
\hline
$\lsumlb svr$ & Lower bound on $S(v)$ if $\reac v=r$, used in the \ublb\ function (see Eq.~\ref{eq:lbound1},~\ref{eq:lbound1_dir})\\
\hline
\hline
\multicolumn{2}{|c|}{\textbf{Farness functions}}\\ 
\hline
$\farn v$ & Farness of node $v$, that is, $\frac{(n-1)S(v)}{(|\Reac v|-1)^2}$\\
\hline
$\lfarn vr$ & Generic lower bound on $\farn v$, if $\reac v=r$ \\
\hline
$\lfarnnb vr$ & Lower bound on $\farn v$, if $\reac v=r$, defined as $(n-1)\frac{\lsumnb vr}{(r-1)^2}$ \\
\hline
$\lfarncut dvr$ & Lower bound on $\farn v$, if $\reac v=r$, defined as $(n-1)\frac{\lsumcut dvr}{(r-1)^2}$ \\
\hline
$\lfarnlb svr$ & Lower bound on $\farn v$, if $\reac v=r$, defined as $(n-1)\frac{\lsumlb svr}{(r-1)^2}$ \\
\hline
\end{tabular}}
\end{footnotesize}
\end{center}
%\vspace{-0.8cm}
\end{table}

We assume the reader to be familiar with the basic notions of graph theory (see, for example,~\cite{Cormen2009}).
Our algorithmic results apply both to undirected and directed graphs. We will make clear in the respective context where results
apply to only one of the two. For example, the hardness results in Section~\ref{sec:comp} apply to directed graphs only.
All the notations and definitions used throughout this paper are summarised in Table~\ref{tbl:notation} (in any case, all notations are also defined in the text). 
Here, let us only define precisely the closeness centrality of a vertex $v$. As already said, in a connected graph, the farness of a node $v$ in a graph $G=(V,E)$ is $\farn v=\frac{\sum_{w \in V} d(v,w)}{n-1}$, and the closeness centrality of $v$ is $\frac{1}{\farn v}$. In the disconnected case, the most natural generalization would be $\farn v=\frac{\sum_{w \in \Reac v} d(v,w)}{\reac v-1}$, and $\clos v=\frac{1}{\farn v}$, where $\Reac v$ is the set of vertices reachable from $v$, and $\reac v=|\Reac v|$. However, this definition does not capture our intuitive notion of centrality: indeed, if $v$ has only one neighbor $w$ at distance $1$, and $w$ has out-degree $0$, then $v$ becomes very central according to this measure, even if $v$ is intuitively peripheral. For this reason, we consider the following generalization, which is quite established in the literature \cite{Lin1976,Wasserman1994,Boldi2013,Boldi2014,Olsen2014}:

\begin{equation}
\farn v=\frac{\sum_{w \in \Reac v} d(v,w)}{\reac v-1}\cdot\frac{n-1}{\reac v-1} \quad\quad \clos{v}=\frac{1}{\farn v}
\end{equation}

If a vertex $v$ has (out)degree $0$, the previous fraction becomes $\frac{0}{0}$: in this case, the closeness of $v$ is set to $0$.

Another possibility is to consider a slightly different definition:
\[
\clos{v}=\sum_{w \in V} f(d(v,w)),
\]
for some decreasing function $f$.\footnote{Usually, it is also assumed without loss of generality that $f(+\infty)=0$, that is, we consider only reachable vertices: if this is not the case, it is enough to use a new function defined by $g(d)=f(d)-f(+\infty)$.} One of the most common choices of $f$ is $f(d)=\frac{1}{d}$: this way, we obtain the harmonic centrality \cite{Marchiori2000}.

In this paper, we focus on Lin's index, because it is quite established in the literature, because the previously best exact top-$k$ closeness centrality algorithm uses this definition~\cite{Olsen2014}, and because, when restricted to the connected case, this definition coincides with closeness centrality (from now on, in a disconnected context, we use closeness centrality to indicate Lin's index). However, all our algorithms can be easily adapted to any centrality measure of the form $\clos{v}=\sum_{w \in V} f(d(v,w))$: indeed, in Sect.~\ref{sec:exp}, we show that our algorithm performs very well also with harmonic centrality.

%Other definitions of closeness centrality exist, in particular reversing the direction in directed graphs (summing over $d(w,v)$ instead of $d(v,w)$). We follow here the definition used by Olsen et al.~\cite{Olsen2014}, the previously best exact top-$k$ closeness centrality algorithm.
%Our algorithmic results should be easily generalizable to most other closeness definitions. Incorporating the reversed direction is trivial: In fact, in our case study on Wikipedia (Section~\ref{sec:wiki}), we already consider both directions by also working on the transposed graph.

\section{Complexity of Computing the Most Central Vertex} \label{sec:comp}

In this section, we show that, even in the computation of the most central vertex, the \emph{textbook} algorithm is almost optimal in the worst case, assuming the Orthogonal Vector conjecture \cite{Williams2005,Abboud2016}, or the well-known Strong Exponential Time Hypothesis (SETH) \cite{Impagliazzo2001}. The Orthogonal Vector conjecture says that, given $N$ vectors in $\{0,1\}^d$, where $d=\O(\log^k N)$ for some $k$, it is impossible to decide if there are two orthogonal vectors in $\O(N^{2-\epsilon})$, for any $\epsilon>0$ not depending on $k$. The SETH says that the $k$-\textsc{Satisfiablility} problem cannot be solved in time $\O((2-\epsilon)^{N})$, where $N$ is the number of variables and $\epsilon$ is a positive constant not depending on $k$. Our reduction is summarized by the following theorem.

%We recall that, if $G$ is not (strongly) connected, the closeness centrality of a vertex $v$ is defined by
%\[\clos v =\frac{(\reac v - 1)^2}{(n-1)\sum_{w \in \Reac v}d(v,w)},\]
%where $\Reac v$ is the set of vertices reachable from $v$, and $\reac v=|\Reac v|$ (note that $v \in \Reac v$ by definition). If a vertex has (out)degree $0$, the previous fraction becomes $\frac{0}{0}$: in this case, the closeness is $0$.

%In this section, we prove that the $\O(mn)$ worst case complexity of the \emph{textbook} algorithm cannot be significantly improved, unless the well-known Strong Exponential Time Hypothesis (SETH) is false \cite{Impagliazzo2001} (this hypothesis says that the $k$-\textsc{Satisfiablility} problem cannot be solved in time $O((2-\epsilon)^{n})$, where $\epsilon>0$ does not depend on $k$).

\begin{theorem}
\label{thm:hardness}
On directed graphs, in the worst case, an algorithm computing the most closeness central vertex in time $\O(m^{2-\epsilon})$ for some $\epsilon>0$ would falsify the Orthogonal Vector conjecture. The same result holds even if we restrict the input to sparse graphs, where $m=\O(n)$.
\end{theorem}
It is worth mentioning that this result still holds if we restrict our analysis to graphs with small diameter. Indeed, the diameter of the graph obtained from the reduction is $9$. Moreover, it is well known that the Orthogonal Vector conjecture is implied by SETH \cite{Williams2005,Borassi2015,Abboud2016}: consequently, the following corollary holds.

\begin{corollary}
\label{cor:hardness}
On directed graphs, in the worst case, an algorithm computing the most closeness central vertex in time $\O(m^{2-\epsilon})$ for some $\epsilon>0$ would falsify SETH. The same result holds even if we restrict the input to sparse graphs, where $m=\O(n)$.
\end{corollary}

The remainder of this section is devoted to the proof of Theorem~\ref{thm:hardness}. We construct a reduction from the $l$-\textsc{TwoDisjointSet} problem, that is, finding two disjoint sets in a collection $\C$ of subsets of a given ground set $X$, where $|X|=\O(\log^l(|\C|))$. For example, $X$ could be the set of numbers between $0$ and $h$, and $\C$ could be the collection of subsets of even numbers between $0$ and $h$ (in this case, the answer is True, since there are two disjoint sets in the collection). It is simple to prove that this problem is equivalent to the Orthogonal Vector problem, by replacing a set $X$ with its characteristic vector in $\{0,1\}^{|X|}$ \cite{Borassi2015}: consequently, an algorithm solving this problem in $\O(|\C|^{2-\epsilon})$ would falsify the Orthogonal Vector conjecture. For a direct reduction between the $l$-\textsc{TwoDisjointSet} problem and SETH, we refer to \cite{Williams2005} (where the \textsc{TwoDisjointSet} problem is named \textsc{CooperativeSubsetQuery}).

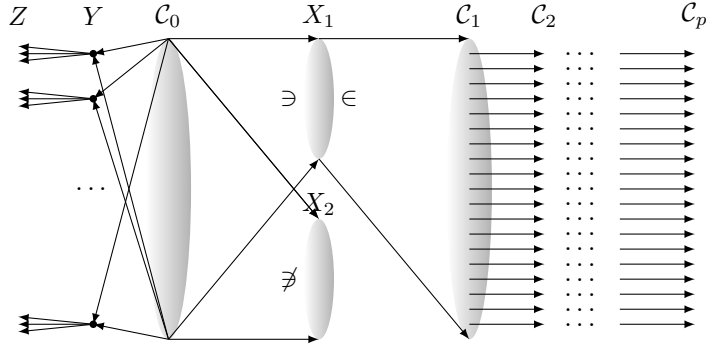
\begin{figure*}[t!b]
\centering
\begin{tikzpicture}

\node at (0,2.3) {$\C_0$};
\node at (2,2.3) {$X_1$};
\node at (2,-0.2) {$X_2$};
\node at (4,2.3) {$\C_1$};
\node at (5,2.3) {$\C_2$};
\node at (7,2.3) {$\C_{p}$};
\node at (-1,2.3) {$Y$};
\node at (-2,2.3) {$Z$};

\draw[draw=none, shading=axis, left color=white, right color=black!30] (0,0) ellipse (.3 and 2);

\draw[draw=none, shading=axis, left color=white, right color=black!30] (2,1.2) ellipse (.2 and .8);

\draw[draw=none, shading=axis, left color=white, right color=black!30] (2,-1.2) ellipse (.2 and .8);

\draw[draw=none, shading=axis, left color=white, right color=black!30] (4,0) ellipse (.3 and 2);

\draw[->, >=latex] (0,2) -- (2,-0.4);
\draw[->, >=latex] (0,-2) -- (2,-2);
\draw[->, >=latex] (0,2) -- (2,-0.4);
\draw[->, >=latex] (0,2) -- (2,2);
\draw[->, >=latex] (0,-2) -- (2,0.4);
\draw[->, >=latex] (2,2) -- (4,2);
\draw[->, >=latex] (2,0.4) -- (4,-2);

\foreach \x in {-1.8,-1.6,...,1.8} {
\draw[->, >=latex] (4,\x) -- (5,\x);
\node[] at (5.5,\x) {$\dots$};
\draw[->, >=latex] (6,\x) -- (7,\x);
}

\foreach \x in {-1.8,1.2,1.8} {
  \node[fill=black,circle,inner sep=0,minimum size=.1cm] at (-1,\x) {};
    \draw[->, >=latex] (0,-2) -- (-1,\x);
    \draw[->, >=latex] (0,2) -- (-1,\x);
    \foreach \y in {-0.1,0,0.1} {
    \draw[->, >=latex] (-1,\x) -- (-2,\x+\y);
    
    }
}

\node at (-1,0) {$\dots$};

\node at (1.6,1.2) {$\ni$};
\node at (1.6,-1.2) {$\not \ni$};
\node at (2.4,1.2) {$\in$};

\end{tikzpicture}
\caption{Reducing the \textsc{TwoDisjointSet} problem to the problem of  finding the most closeness central vertex.}
\label{fig:reduction}
\end{figure*}

Given an instance $(X,\C)$ of the $l$-\textsc{TwoDisjointSet} problem, and given a set $C \in \C$, let $R_C$ be $|\{C' \in \C:C \cap C' \neq \emptyset\}|$. The \textsc{TwoDisjointSet} problem has no solutions if and only if $R_C=|\C|$ for all $C \in \C$; indeed, $R_C=|\C|$ means that $C$ intersects all the sets in $\C$. We construct a directed graph $G=(V,E)$, where $|V|,|E|=\O(|\C||X|)=\O(|\C|\log^l|\C|)$, such that:
\begin{enumerate}
\item $V$ contains a set of vertices $\C_0$ representing the sets in $\C$ (from now on, if $C \in \C$, we denote by $C_0$ the corresponding vertex in $\C_0$);
\item the centrality of $C_0$ is a function $\clos{R_C}$, depending only on $R_C$ (that is, if $R_C=R_{C'}$ then $\clos{C_0}=\clos {C_0'}$); \label{item:dep}
\item the function $\clos{R_C}$ is decreasing with respect to $R_C$; \label{item:inc}
\item the most central vertex is in $\C_0$. \label{item:big}
\end{enumerate}
In such a graph, the vertex with maximum closeness corresponds to the set $S$ minimizing $R_S$: indeed, it is in $\C_0$ by Condition~\ref{item:big}, and it minimizes $R_S$ by Condition~\ref{item:dep}-\ref{item:inc}. Hence, assuming we can find $S_0$ in time $\O(n^{2-\epsilon})$, we can easily check if the closeness of $S_0$ is $\clos{|\C|}$: if it is not, it means that the corresponding \textsc{TwoDisjointSet} instance has a solution of the form $(S,S_1)$ because $R_S \neq \C$. Otherwise, for each $C$, $R_C \geq R_S = |\C|$, because $c(C_0) \leq c(S_0)=c(|\C|)$, and $c$ is decreasing with respect to $R_C$. This means that $R_C=|\C|$ for each $C$, and there are no two disjoints sets. This way, we can solve the $l$-\textsc{TwoDisjointSet} problem in $\O(n^{2-\epsilon})=\O((|\C|\log^l|\C|)^{2-\epsilon})=\O(|\C|^{2-\frac{\epsilon}{2}})$, against the Orthogonal Vector conjecture, and SETH. If we also want the graph to be sparse, we can add $\O(|\C|\log^l|\C|)$ nodes with no outgoing edge.

To construct this graph (see Figure~\ref{fig:reduction}), we start by adding to $V$ the copy $\C_0$ of $\C$, another copy $\C_1$ of $\C$ and a copy $X_1$ of $X$. These vertices are connected as follows: for each element $x \in X$ and set $C \in \C$, we add an edge $(C_0,x)$ and $(x,C_1)$, where $C_0$ is the copy of $C$ in $\C_0$, and $C_1$ is the copy of $C$ in $\C_1$. Moreover, we add a copy $X_2$ of $X$ and we connect all pairs $(C_0,x)$ with $C \in \C$, $x \in X$ and $x \notin C$. This way, the closeness centrality of a vertex $C_0 \in \C_0$ is $\frac{(|X|+R_C)^2}{(n-1)(|X|+2R_C)}$ (which only depends on $R_C$).
%This way, if $C \in \C_0$, the closeness centrality of $C$ is $\frac{(|C|+R_C)^2}{(n-1)(|C|+2R_C)}$ because, from $C$, it is possible to reach $|C|$ vertices in $X_1$, at distance $1$, and $R_C$ vertices in $X_2$, at distance $2$. To remove the dependency from $|C|$, we add a copy $X_2$ of $X$ and we connect all pairs $(C,x)$ with $C \in \C_0$, $x \in X_2$ and $x \notin C$: this way, the closeness centrality of $C$ becomes $\frac{(|X|+R_C)^2}{(n-1)(|X|+2R_C)}$ (which only depends on $R_C$). 
To enforce Conditions~\ref{item:inc}-\ref{item:big}, we add a path of length $p$ leaving each vertex in $\C_1$, and $q$ vertices linked to each vertex in $\C_0$, each of which has out-degree $|\C|$: we show that by setting $p=7$ and $q=36$, all required conditions are satisfied.

More formally, we have constructed the following graph $G=(V,E)$:
\begin{itemize}
\item $V=Z \cup Y \cup \C_0 \cup X_1 \cup X_2 \cup \C_1 \cup \dots \cup \C_p$, where $Z$ is a set of cardinality $q|\C|$, $Y$ a set of cardinality $q$, the $\C_i$s are copies of $\C$ and the $X_i$s are copies of $X$;
\item each vertex in $Y$ has $|\C|$ neighbors in $Z$, and these neighbors are disjoint;
\item for each $x \in C$, there are edges from $C_0 \in \C_0$ to $x \in X_1$, and from $x \in X_1$ to $C_1 \in \C_1$;
\item for each $x \notin C$, there is an edge from $C_0 \in \C_0$ to $x \in X_2$;
\item each $C_i \in \C_i$, $1 \leq i \leq p$, is connected to the same set $C_{i+1} \in \C_{i+1}$;
\item no other edge is present in the graph.
\end{itemize}
Note that the number of edges in this graph is $\O(|\C||X|)=\O(|\C|\log^l(|\C|))$, because $|X|<\log^l(|\C|)$, 

\begin{lemma}
Assuming $|\C|>1$, all vertices outside $\C_0$ have closeness centrality at most $\frac{2|\C|}{n-1}$, where $n$ is the number of vertices.
\end{lemma}
\begin{proof}
If a vertex is in $Z, X_2$, or $\C_p$, its closeness centrality is not defined, because it has out-degree $0$. 

A vertex $y \in Y$ reaches $|\C|$ vertices in $1$ step, and hence its closeness centrality is $\frac{|\C|^2}{|\C|(n-1)}=\frac{|\C|}{n-1}$.

A vertex in $\C_i$ reaches $p-i$ other vertices, and their distance is $1,\dots, p-i$: consequently, its closeness centrality is $\frac{(p-i)^2}{\frac{(p-i)(p-i+1)}{2}(n-1)}=\frac{2(p-i)}{(n-1)(p-i+1)} \leq \frac{2}{n-1}$. 

Finally, for a vertex $x \in X_1$ contained in $N_x$ sets, for each $1 \leq i \leq p$, $x$ reaches $N_x$ vertices in $\C_i$, and these vertices are at distance $i$. Hence, the closeness of $x$ is $\frac{(pN_x)^2}{\frac{p(p+1)}{2}N_x(n-1)}=\frac{2pN_x}{(n-1)(p+1)}\leq \frac{2N_x}{n-1}\leq \frac{2|\C|}{n-1}$. This concludes the proof.
\qed 
\end{proof}

Let us now compute the closeness centrality of a vertex $C \in \C_0$. The reachable vertices are:
\begin{itemize}
\item all $q$ vertices in $Y$, at distance $1$;
\item all $|\C|q$ vertices in $Z$, at distance $2$;
\item $|X|$ vertices in $X_1$ or $X_2$, at distance $1$;
\item $R_C$ vertices in $\C_i$ for each $i$, at distance $i+1$ (the sum of the distances of these vertices is $\sum_{i=1}^p i+1=-1+\sum_{i=1}^{p+1} i=\frac{(p+2)(p+1)}{2}-1$).
\end{itemize}
Hence, the closeness centrality of $C$ is:
\begin{align*}
c(R_C)&=\frac{(q(1+|\C|)+|X|+pR_C)^2}{\left(q(1+2|\C|)+|X|+\left(\frac{(p+1)(p+2)}{2}-1\right)R_C\right)(n-1)} \\
&=\frac{(q(1+|\C|)+|X|+pR_C)^2}{\left(q(1+2|\C|)+|X|+g(p)R_C\right)(n-1)}
\end{align*}
where $g(p)=\frac{(p+1)(p+2)}{2}-1$. We want to choose $p$ and $q$ verifying:
\begin{enumerate}[label=\alph*.]
\item the closeness of vertices in $\C_0$ is bigger than $\frac{2|\C|}{n-1}$ (and hence bigger than the closeness of all other vertices); \label{item:bigger}
\item $c(R_C)$ is a decreasing function of $R_C$ for $0\leq R_C \leq |\C|$. \label{item:decr}
\end{enumerate}

In order to satisfy Condition~\ref{item:decr}, the derivative $c'(R_C)$ of $c$ is
\noindent
\begin{small}
$
%c'(R_C)=\frac{2p(q(1+|\C|)+|X|+pR_C)\left(q(1+2|\C|)+|X|+g(p)R_C\right)-g(p)(q(1+|\C|)+|X|+pR_C)^2}{\left(q(1+2|\C|)+|X|+g(p)R_C\right)^2(n-1)} \\
(q(1+|\C|)+|X|+pR_C)\frac{[pg(p)R_c+2p\left(q(1+2|\C|)+|X|\right)-g(p)(q(1+|\C|)+|X|)]}{\left(q(1+2|\C|)+|X|+g(p)R_C\right)^2(n-1)}.
$
\end{small}

This latter value is negative if and only if $pg(p)R_c+2p\left(q(1+2|\C|)+|X|\right)-g(p)(q(1+|\C|)+|X|)<0$. Assuming $g(p) \geq 5p$ and $R_C<|\C|$, this value is:
\begin{small}
\begin{align*}
&pg(p)R_C+2p\left(q(1+2|\C|)+|X|\right)-g(p)(q(1+|\C|)+|X|) \\
&\leq pg(p)|\C|+2pq+4pq|\C|+2p|X|-g(p)(q-|C|-|X|) \\
&\leq pg(p)|\C|+4pq|\C|-g(p)q|\C| \\
&\leq pg(p)|\C|-pq|\C|.
\end{align*}
\end{small}
Assuming $q>g(p)$, we conclude that $c'(R_C)<0$ for $0 \leq R_C \leq |\C|$, and we verify Condition~\ref{item:decr}. In order to verify Condition~\ref{item:bigger}, we want $c(R_C)\geq \frac{2|\C|}{n+1}$ (since $c(R_C)$ is decreasing, it is enough $c(|\C|)\geq \frac{2|\C|}{n+1}$). Under the assumptions $q>g(p)$, $0<|X| \leq |\C|$ (which trivially holds for $|\C|$ big enough, because $|X| \leq \log^p|\C|$),
\begin{align*}
c(|\C|)&=\frac{(q(1+|\C|)+|X|+pR_C)^2}{\left(q(1+2|\C|)+|X|+g(p)R_C\right)(n-1)} \\
&\geq \frac{q^2|\C|^2}{(q(3|\C|)+|\C|+|\C|)(n-1)} \\
& \geq \frac{q|\C|}{5(n-1)} > \frac{2|\C|}{n-1}
\end{align*}
if $q>10$. 

To fulfill all required conditions, it is enough to choose $p=7$, $g(p)=35$, and $q=36$.

\section{Overview of the Algorithm}
\label{sec:overview}
In this section, we describe our new approach for computing the $k$ nodes with maximum closeness (equivalently, the $k$ nodes with minimum farness, where the farness $\farn v$ of a vertex $v$ is $\frac{1}{\clos v}=\frac{(n-1)\sum_{w \in \Reac v} d(v,w)}{(\reac v-1)^2}$, as in Table~\ref{tbl:notation}). If we have more than one node with the same score, we output all nodes having a centrality bigger than or equal to the centrality of the $k$-th node. 

In the previous section, we have shown that the trivial algorithm cannot be improved in the worst case: here, we describe an algorithm that is much more efficient when tested on real-world graphs. The basic idea is to keep track of a lower bound on the farness of each node, and to skip the analysis of a vertex $v$ if this lower bound implies that $v$ is not in the top $k$.

More formally, let us assume that we know the farness of some vertices $v_1,\dots,v_l$, and a lower bound $\lfarns(w)$ on the farness of any other vertex $w$. Furthermore, assume that there are $k$ vertices among $v_1,\dots,v_l$ verifying $\farn {v_i} > \lfarns(w) ~ \forall w \in V-\{v_1,\dots,v_l\}$, and hence $\farn{w} \leq \lfarns(w) < \farn w ~ \forall w \in V-\{v_1,\dots,v_l\}$. Then, we can safely skip the exact computation of $\farn w$ for all remaining nodes $w$, because the $k$ vertices with smallest farness are among $v_1,\dots,v_l$.

This idea is implemented in Algorithm~\ref{alg:main}: we use a list \texttt{Top} containing all ``analysed'' vertices $v_1,\dots,v_l$ in increasing order of farness, and a priority queue \texttt{Q} containing all vertices ``not analysed, yet'', in increasing order of lower bound $\lfarns$ (this way, the head of \texttt{Q} always has the smallest value of $\lfarns$ among all vertices in \texttt{Q}). At the beginning, using the function $\cb()$, we compute a first bound $\lfarns(v)$ for each vertex $v$, and we fill the queue \texttt{Q} according to this bound. Then, at each step, we extract the first element $v$ of \texttt{Q}: if $\lfarns(v)$ is smaller than the $k$-th biggest farness computed until now (that is, the farness of the $k$-th vertex in variable \texttt{Top}), we can safely stop, because for each $x \in \texttt{Q}$, $\farn x \leq \lfarns(x) \leq \lfarns(v) < \farn{\texttt{Top[}k\texttt{]}}$, and $x$ is not in the top $k$. Otherwise, we run the function $\ub(v)$, which performs a BFS from $v$, returns the farness of $v$, and improves the bounds $\lfarns$ of all other vertices. Finally, we insert $v$ into \texttt{Top} in the right position, and we update \texttt{Q} if the lower bounds have changed.

\begin{algorithm2e}
\begin{footnotesize}
\LinesNumbered
\SetKwFunction{BFS}{updateBounds}
\SetKwFunction{computeBounds}{computeBounds}
\SetKwFunction{extractMin}{extractMin}
\SetKwData{farn}{Farn}
\SetKwData{Q}{Q}
\SetKwData{L}{L}
\SetKwData{Top}{Top}
 \SetKwInOut{Input}{Input}
 \SetKwInOut{Output}{Output}
\Input{A graph $G=(V,E)$}
\Output{Top $k$ nodes with highest closeness and their closeness values $c(v)$}
global $\L,\Q \gets \computeBounds{G}$\; \label{line:lbounds}
global $\Top \leftarrow [\ ]$\;
global $\farn$\;
\lFor{$v \in V$}{$\farn[v]=+\infty$}
\While{$\Q$ is not empty} {
    $v \gets \Q.\extractMin()$\;
    \lIf{$|\Top|\geq k$ and $\L[v] > \Top[k]$}{\Return{$\Top$}}
    $\farn[v] \gets \BFS(v)$; // This function might also modify $\L$ \\
    add $v$ to \Top, and sort \Top according to \farn\;
    update $\Q$ according to the new bounds\;
}
\end{footnotesize}
\caption{Pseudocode of our algorithm for top-$k$ closeness centralities.}
\label{alg:main}
\end{algorithm2e}

The crucial point of the algorithm is the definition of the lower bounds, that is, the definition of the functions \cb\ and \ub. We propose two alternative strategies for each of these these two functions: in both cases, one strategy is conservative, that is, it tries to perform as few operations as possible, while the other strategy is aggressive, that is, it needs many operations, but at the same time it improves many lower bounds.

Let us analyze the possible choices of the function \cb. The conservative strategy \cbdeg\ needs time $\O(n)$: it simply sets $\lfarns(v)=0$ for each $v$, and it fills \texttt{Q} by inserting nodes in decreasing order of degree (the idea is that vertices with high degree have small farness, and they should be analysed as early as possible, so that the values in \texttt{Top} are correct as soon as possible). Note that the vertices can be sorted in time $\O(n)$ using counting sort.

The aggressive strategy \cbnb\ needs time $\O(mD)$, where $D$ is the diameter of the graph: it computes the neighborhood-based lower bound $\lfarnnbr v$ for each vertex $v$ (we will explain shortly afterwards how it works), it sets $\lfarns(v)=\lfarnnbr v$, and it fills \texttt{Q} by adding vertices in decreasing order of $\lfarns$. The idea behind the neighborhood-based lower bound is to count the number of paths of length $l$ starting from a given vertex $v$, which is also an upper bound $U_l$ 
on the number of vertices at distance $l$ from $v$. From $U_l$, it is possible to define a lower bound on $\sum_{x \in V}d(v,x)$ by ``summing $U_l$ times the distance $l$'', until we have summed $n$ distances: this bound yields the desired lower bound on the farness of $v$. The detailed explanation of this function is provided in Sect.~\ref{sec:nblb}.

For the function \ub$(w)$, the conservative strategy \ubcut$(w)$ does not improve $\lfarns$, and it cuts the BFS as soon as it is sure that the farness of $w$ is smaller than the $k$-th biggest farness found until now, that is, $\texttt{Farn}[\texttt{Top}[k]]$. If the BFS is cut, the function returns $+\infty$, otherwise, at the end of the BFS we have computed the farness of $v$, and we can return it. The running time of this procedure is $\O(m)$ in the worst case, but it can be much better in practice. It remains to define how the procedure can be sure that the farness of $v$ is at least $x$: to this purpose, during the BFS, we update a lower bound on the farness of $v$. The idea behind this bound is that, if we have already visited all nodes up to distance $d$, we can upper bound the closeness centrality of $v$ by setting distance $d+1$ to a number of vertices equal to the number of edges ``leaving'' level $d$, and distance $d+2$ to all the remaining vertices. The details of this procedure are provided in Sect.~\ref{sec:bfscut}.

The aggressive strategy \ublb$(v)$ performs a complete BFS from $v$, and it bounds the farness of each node $w$ using the level-based lower bound. The running time is $\O(m)$ for the BFS, and $\O(n)$ to compute the bounds. The idea behind the level-based lower bound is that $d(w,x) \geq |d(v,w)-d(v,x)|$, and consequently $\sum_{x \in V} d(w,x) \geq \sum_{x \in V} |d(v,w)-d(v,x)|$. The latter sum can be computed in time $\O(n)$ for each $w$, because it depends only on the level $d$ of $w$ in the BFS tree, and because it is possible to compute in $\O(1)$ the sum for a vertex at level $d+1$, if we know the sum for a vertex at level $d$. The details are provided in Sect.~\ref{sec:lblb}.

Finally, in order to transform these lower bounds on $\sum_{x \in V} d(v,x)$ into bounds on $\farn v$, we need to know the number of vertices reachable from a given vertex $v$. In Sect.~\ref{sec:nblb}, \ref{sec:bfscut}, \ref{sec:lblb}, we assume that these values are known: this assumption is true in undirected graphs, where we can compute the number of reachable vertices in linear time at the beginning of the algorithm, and in strongly connected directed graphs, where the number of reachable vertices is $n$. The only remaining case is when the graph is directed and not strongly connected: in this case, we need some additional machinery, which are presented in Sect.~\ref{sec:disc}.

\section{Neighborhood-Based Lower Bound}
\label{sec:nblb}
In this section, we propose a lower bound $\lsumnb{v}{\reac v}$ on the total sum $S(v)=\sum_{w \in \Reac v} d(v,w)$ of an undirected or strongly-connected graph. If we know the number $\reac v$ of vertices reachable from $v$, this bound translates into a lower bound on the farness of $v$, simply multiplying by $(n-1)/(r(v)-1)^2$. The basic idea is to find an upper bound $\tsneig{i}{v}$ on the number of nodes $\sneig{i}{v}$ at distance $i$ from $v$. Then, intuitively, if we assume that the number of nodes at distance $i$ is greater than its actual value and ``stop counting'' when we have $\reac v$ nodes, we get something that is smaller than the actual total distance. This is because we are assuming that the distances of some nodes are smaller than their actual values. This argument is formalized in Prop.~\ref{prop:lbound}.
\begin{proposition}
\label{prop:lbound}
If $\tsneig{i}{v}$ is an upper bound on $\sneig{i}{v}$, for $i = 0, ..., \diam(G)$ and $\ecc(v) := \max_{w \in \reac v} d(v,w)$, then $\lsumnb{v}{\reac v} := \sum_{k=1}^{\ecc(v)} k \cdot \min \left \lbrace \tsneig{k}{v},\  \max \left \lbrace \reac v - \sum_{i=0}^{k-1} \tsneig{i}{v},\ 0 \right \rbrace \right \rbrace$ is a lower bound on $S(v)$.
\end{proposition}
\begin{proof}
First, we notice that $S(v) = \sum_{k=0}^{\ecc(v)} k \cdot \sneig{k}{v}$ and $\reac v = \sum_{k=0}^{\ecc(v)} \sneig{k}{v}$.

Let us assume that $\tsneig{0}{v} < \reac v$. In fact, if $\tsneig{0}{v} \geq \reac v$, the statement is trivially satisfied. Then, there must be a number $\ecc'>0$ such that for $k<\ecc'$ the quantity $\min \left \lbrace \tsneig{k}{v},\  \max \left \lbrace \reac v - \sum_{i=0}^{k-1} \tsneig{i}{v},\ 0 \right \rbrace \right \rbrace$ is equal to $\tsneig{k}{v}$, for $k= \ecc'$, the quantity is equal to $\alpha := \reac v - \sum_{k=0}^{\ecc'-1} \tsneig{k}{v} > 0$ and, for $k > \ecc'$, it is equal to 0. Therefore we can write $\lsumnb{v}{\reac v}$ as $\sum_{k=1}^{\ecc'-1}k\cdot \tsneig{k}{v}+ \ecc'\cdot \alpha$.

We show that $\ecc' \leq \ecc(v)$. In fact, we know that $\sum_{k=0}^{\ecc'-1}\tsneig{k}{v} <\reac v = \sum_{k=0}^{\ecc(v)} \sneig{k}{v} \leq \sum_{k=0}^{\ecc(v)}\tsneig{k}{v}$. Therefore $\ecc'-1 < \ecc(v)$, which implies $\ecc' \leq \ecc(v)$.

For each $i$, we can write $\tsneig{i}{v} = \sneig{i}{v} + \epsilon_i$, $\epsilon_i \geq 0$. Therefore, we can write $\sum_{k=0}^{\ecc'-1} \epsilon_i + \alpha = \reac v - \sum_{k=0}^{\ecc'-1} \sneig{k}{v} = \sum_{k=\ecc'}^{\ecc(v)} \sneig{k}{v}$. Then, $\lsumnb{v}{\reac v} = \sum_{k=0}^{\ecc'-1}k \cdot \sneig{k}{v} + \sum_{k=0}^{\ecc'-1} k \cdot \epsilon_i + \ecc'\cdot \alpha \leq \sum_{k=0}^{\ecc'-1} k \cdot \sneig{k}{v} + \ecc'(\alpha + \sum_{k=0}^{\ecc'-1}\epsilon_i) = \sum_{k=0}^{\ecc'-1}k \cdot \sneig{k}{v} + \ecc'(\sum_{k=\ecc'}^{\ecc(v)} \sneig{k}{v}) \leq \sum_{k=0}^{\ecc(v)} k \cdot \sneig{k}{v} = S(v)$.
\qed
\end{proof}

In the following paragraphs, we propose upper bounds $\tsneig{i}{v}$ for trees, undirected graphs and directed strongly-connected graphs. In case of trees, the bound $\tsneig{i}{v}$ is actually equal to $\sneig{i}{v}$, which means that the algorithm can be used to compute closeness of all nodes in a tree exactly.

\paragraph{Computing closeness on trees} 
Let us consider a node $s$ for which we want to compute the total distance $S(s)$ (notice that in a tree $\clos{s} = (n-1)/S(s)$). The number of nodes at distance 1 in the BFS tree from $s$ is clearly the degree of $s$. What about distance 2? Since there are no cycles, all the neighbors of the nodes in $\neig 1s$ are nodes at distance 2 from $s$, with the only exception of $s$ itself. Therefore, naming $\neig k s$ the set  of nodes at distance $k$ from $s$ and $\sneig k s$ the number of these nodes, we can write $ \sneig 2 s = \sum_{w \in \neig 1s}  \sneig 1 w - \degree(s)$. In general, we can always relate the number of nodes at each distance $k$ of $s$ to the number of nodes at distance $k-1$ in the BFS trees of the neighbors of $s$.
Let us now consider $ \sneig k s$, for $k>2$.
\begin{figure*}[h!tb]
\begin{center}
\includegraphics[width = 0.7\textwidth]{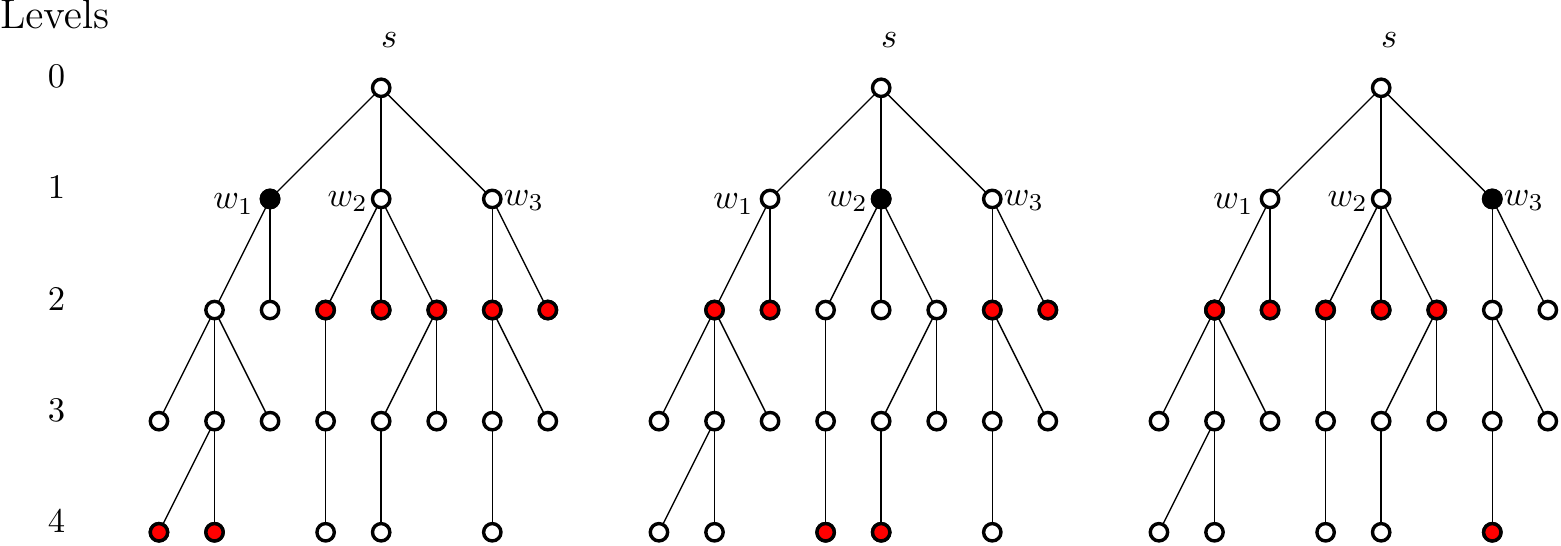}
\caption{Relation between nodes at distance 4 for $s$ and the neighbors of $s$. The red nodes represent the nodes at distance 3 for $w_1$ (left), for $w_2$ (center) and for $w_3$ (right).}
\label{fig:tree}
% \vspace{-4ex}
\end{center}
\end{figure*}
Figure~\ref{fig:tree} shows an example where $s$ has three neighbors $w_1$, $w_2$ and $w_3$. Suppose we want to compute $\neig 4s$ using information from $w_1$, $w_2$ and $w_3$. Clearly, $ \neig 4 s \subset \neig 3 {w_1} \cup \neig 3 {w_2} \cup \neig 3 {w_3}$; however, there are also other nodes in the union that are not in  $\neig 4 s$. Furthermore, the nodes in $\neig 3 {w_1}$ (red nodes in the leftmost tree) are of two types: nodes in $\neig 4 {s}$ (the ones in the subtree of $w_1$) and nodes in $\neig 2 {s}$ (the ones in the subtrees of $w_2$ and $w_3$). An analogous behavior can be observed for $w_2$ and $w_3$ (central and rightmost trees). If we simply sum all the nodes in $\sneig 3 {w_1}$, $\sneig 3 {w_2}$ and $\sneig 3 {w_3}$, we would be counting each node at level 2 twice, \ie once for each node in $\neig 1s$ minus one. Hence, for each $k>2$, we can write 
\begin{equation}
\label{eq:neighbors}
\sneig k {s}= \sum_{w \in \neig 1s} \sneig {k-1} {w} - \sneig {k-2} {s}\cdot (\degree(s)-1).
\end{equation}

\begin{algorithm2e}
 \begin{footnotesize}
\LinesNumbered
\SetKwFunction{BFSfrom}{BFSfrom}
\SetKwFunction{updateLowerBounds}{updateLowerBounds}
\SetKwFunction{enqueue}{enqueue}
\SetKwFunction{extractMin}{extractMin}
\SetKwData{neighborss}{$\sneig{k-1}{s}$}
\SetKwData{newNeighborss}{$\sneig{k}{s}$}
\SetKwData{neighborsw}{$\sneig{k-1}{w}$}
\SetKwData{newNeighborsw}{$\sneig{k}{w}$}
\SetKwData{level}{k}
\SetKwData{nFinished}{nFinished}
\SetKwData{oldNeighborss}{$\sneig{k-2}{s}$}
\SetKwData{oldNeighborsw}{$\sneig{k-2}{w}$}
\SetKwData{S}{$S$}
 \SetKwInOut{Input}{Input}
 \SetKwInOut{Output}{Output}
\Input{A tree $T=(V,E)$}
\Output{Closeness centralities $c(v)$ of each node $v\in V$}
\ForEach{$s \in V$}{
	\label{line:initloop}
	$\neighborss \leftarrow \degree(s)$\;
	$\S(s) \leftarrow \degree(s)$\;
	\label{line:initloop_end}
}
$\level \leftarrow 2$\;
$\nFinished \leftarrow 0$\;
\While{$\nFinished < n$} {
	\label{line:while}
	\ForEach{$s \in V$} {
		\If{$\level = 2$ }{
			$\newNeighborss \leftarrow \sum_{w \in N(s)} \neighborsw - \degree(s)$\; \label{line:level2}
		}\Else {
			$\newNeighborss \leftarrow \sum_{w \in N(s)} \neighborsw - \oldNeighborss (\degree(s)-1)$\; \label{line:levelgt2}
		}
	}
	\ForEach{$s \in V$} {
		$\oldNeighborss \leftarrow \neighborss$\;
		$\neighborss \leftarrow \newNeighborss$\;
		\If{$\neighborss > 0$} {
			$\S(s) \leftarrow \S(s) + \level \cdot \neighborss $\;
		} \Else {
			$\nFinished \leftarrow \nFinished + 1$\;
		} 
	}
	$\level \leftarrow \level + 1$\;
	\label{line:while_end}
}
\ForEach{$s \in V$} {
	$c(v) \leftarrow (n-1)/\S(v)$\;
}
\Return{$c$}
\end{footnotesize}
\caption{Closeness centrality in trees}
\label{algo:trees}
\end{algorithm2e}
\noindent From this observation, we define a new method to compute the total distance of all nodes, described in Algorithm~\ref{algo:trees}. Instead of computing the BFS tree of each node one by one, at each step we compute the number $\sneig{k}{v}$ of nodes at level $k$ for \textit{all} nodes $v$. First (Lines~\ref{line:initloop}~-~\ref{line:initloop_end}), we compute $\sneig{1}{v}$ for each node (and add that to $S(v)$). Then (Lines~\ref{line:while}~-~\ref{line:while_end}), we consider all the other levels $k$ one by one. For each $k$, we use $\sneig{k-1}{w}$ of the neighbors $w$ of $v$ and $\sneig{k-2}{v}$ to compute $\sneig{k}{v}$ (Line~\ref{line:level2} and~\ref{line:levelgt2}). 
If, for some $k$, $\sneig{k}{v} = 0$, all the nodes have been added to $S(v)$. Therefore, we can stop the algorithm when $\sneig{k}{v} = 0 \ \ \forall v \in V$.
\begin{proposition}
\label{prop:tree}
Algorithm~\ref{algo:trees} requires $O(D \cdot n)$ operations to compute the closeness centrality of all nodes in a tree $T$.
\end{proposition}
\begin{proof}
The for loop in Lines~\ref{line:initloop}~-~\ref{line:initloop_end} of Algorithm~\ref{algo:trees} clearly takes $O(n)$ time. For each level of the while loop of Lines~\ref{line:while}~-~\ref{line:while_end}, each node scans its neighbors in Line~\ref{line:level2} or Line~\ref{line:levelgt2}. In total, this leads to $O(n)$ operations per level since $m = O(n)$. Since the maximum number of levels that a node can have is equal to the diameter of the tree, the algorithm requires $O(D \cdot n)$ operations.
\qed
\end{proof}

Note that closeness centrality on trees could even be computed in time $O(n)$ in a different manner~\cite{DBLP:conf/stacs/BrandesF05}. We choose to include Algorithm~\ref{algo:trees} 
here nonetheless since it paves the way for an algorithm computing a lower bound in general undirected graphs, described next.

\paragraph{Lower bound for undirected graphs} For general undirected graphs, Eq.~(\ref{eq:neighbors}) is not true anymore -- but a related upper bound $\tsneig{k}{\cdot}$ on $\sneig{k}{\cdot}$
is still useful. Let $\tsneig{k}{s}$ be defined recursively as in Eq.~(\ref{eq:neighbors}): in a tree, $\tsneig{k}{s}=\sneig{k}{s}$, while in this case we prove that $\tsneig{k}{s}$ is an upper bound on $\neig ks$.
Indeed, there could be nodes $x$ for which there are multiple paths between $s$ and $x$ and that are therefore contained in the subtrees of more than one neighbor of $s$. This means that we would count $x$ multiple times when considering $\tsneig{k}{s}$, overestimating the number of nodes at distance $k$. However, we know for sure that at level $k$ there cannot be \textit{more nodes} than in Eq.~(\ref{eq:neighbors}). If, for each node $v$, we assume that the number $\tsneig{k}{v}$ of nodes at distance $k$ is that of Eq.~(\ref{eq:neighbors}), we can apply Prop.~\ref{prop:lbound} and get a lower bound $\lsumnb{v}{\reac v}$ on the total sum for undirected graphs.
The procedure is described in Algorithm~\ref{algo:lbound2}. The computation of $\lsumnb{v}{\reac v}$ works basically like Algorithm~\ref{algo:trees}, with the difference that here we keep track of the number of the nodes found in all the levels up to $k$ (\textsf{nVisited}) and stop the computation when \textsf{nVisited} becomes equal to $\reac v$ (if it becomes larger, in the last level we consider only $\reac v - \textsf{nVisited}$ nodes, as in Prop.~\ref{prop:lbound} (Lines~\ref{line:else1}~-~\ref{line:else1-end}).
\begin{proposition}
\label{prop:graph}
For an undirected graph $G$, computing the lower bound $\lsumnb{v}{\reac v}$ described in Algorithm~\ref{algo:lbound2} takes $O(D \cdot m)$ time. 
\end{proposition}
\begin{proof}
Like in Algorithm~\ref{algo:trees}, the number of operations performed by Algorithm~\ref{algo:lbound2} at each level of the while loop is $\O(m)$. At each level $i$, all the nodes at distance $i$ are accounted for (possibly multiple times) in Lines~\ref{line:level2lbound} and~\ref{line:levelgt2lbound}. Therefore, at each level, the variable $\mathsf{nVisited}$ is always greater than or equal to the the number of nodes $v$ at distance $d(v) \leq i$. Since $d(v) \leq D$ for all nodes $v$, the maximum number of levels scanned in the while loop cannot be larger than $D$, therefore the total complexity is $O(D \cdot m)$.
\qed
\end{proof}

\smallskip 
\begin{algorithm2e}
 \begin{footnotesize}
\LinesNumbered
\SetKwFunction{BFSfrom}{BFSfrom}
\SetKwFunction{updateLowerBounds}{updateLowerBounds}
\SetKwFunction{enqueue}{enqueue}
\SetKwFunction{extractMin}{extractMin}
\SetKwData{neighborss}{$\sneig{k-1}{s}$}
\SetKwData{newNeighborss}{$\sneig{k}{s}$}
\SetKwData{neighborsw}{$\sneig{k-1}{w}$}
\SetKwData{oldNeighborss}{$\sneig{k-2}{s}$}
\SetKwData{oldNeighborsw}{$\sneig{k-2}{w}$}
\SetKwData{level}{$k$}
\SetKwData{nFinished}{nFinished}
\SetKwData{nVisited}{nVisited}
\SetKwData{finished}{finished}
\SetKwData{S}{$\tilde{S}^{\mathrm{(un)}}$}
 \SetKwInOut{Input}{Input}
 \SetKwInOut{Output}{Output}
\Input{A graph $G=(V,E)$}
\Output{Lower bounds $\lfarnnb v{\reac v}$ of each node $v\in V$}
\ForEach{$s \in V$}{
	$\neighborss \leftarrow \degree(s)$\;
	$\S(s) \leftarrow \degree(s)$\;
	$\nVisited[s] \leftarrow \degree(s) + 1$\;
	$\finished[s] \leftarrow false$\;
}
$\level \leftarrow 2$\;
$\nFinished \leftarrow 0$\;

\While{$\nFinished < n$} {
	\ForEach{$s \in V$} {
		\If{$\level = 2$ }{
			$\newNeighborss \leftarrow \sum_{w \in N(s)} \neighborsw - \degree(s)$\; \label{line:level2lbound}
		}\Else {
			$\newNeighborss \leftarrow \sum_{w \in N(s)} \neighborsw - \oldNeighborss (\degree(s)-1)$\; \label{line:levelgt2lbound}
		}
	}
	\ForEach{$s \in V$} {
		\If{$\finished[v]$}
		{
			\textbf{{continue}}\;
		}
		$\oldNeighborss \leftarrow \neighborss$\;
		$\neighborss \leftarrow \newNeighborss$\;
		$\nVisited[s] \leftarrow \nVisited[s] + \neighborss $\;
		\If{$\nVisited[s] < \reac v$} {  \label{line:if1}
			$\S(s) \leftarrow \S(s) + \level \cdot \neighborss $\;
		} \Else {\label{line:else1}
			$\S(s) \leftarrow \S(s) + \level (\reac v-(\nVisited[s]-\neighborss)) $\;
			$\nFinished \leftarrow \nFinished + 1$\;
			$\finished[s] \leftarrow true$\;\label{line:else1-end}
		} 
	}
	$\level \leftarrow \level + 1$\;
}
\ForEach{$v \in v$} {
$\lfarnnb v{\reac v} \leftarrow \frac{(n-1)\tilde{S}^{\mathrm{(un)}}}{(\reac v-1)^2}$\;
}
\Return{$\lfarnnb \cdot{\reac \cdot}$}
\end{footnotesize}
\caption{Neighborhood-based lower bound for undirected graphs}
\label{algo:lbound2}
\end{algorithm2e}

\paragraph{Lower bound on directed graphs} In directed graphs, we can simply consider the out-neighbors, without subtracting the number of nodes discovered in the subtrees of the other neighbors in Eq.~(\ref{eq:neighbors}). The lower bound (which we still refer to as $\lsumnb{v}{\reac v}$) is obtained by replacing Eq.~(\ref{eq:neighbors}) with the following in Lines~\ref{line:level2lbound} and~\ref{line:levelgt2lbound} of Algorithm~\ref{algo:lbound2}: 
\begin{equation}
\label{eq:neighbors_dir}
\tsneig{k}{s} = \sum_{w \in \Gamma(s)} \tsneig{k-1}{w}
\end{equation}
%A lower bound on the farness $\farn{v}$ can be easily computed as $\frac{(n-1)\lsumnb{v}{\reac v}}{(\reac v-1)^2}$ in both undirected and strongly-connected graphs. Indeed, in the first case the number of reachable nodes $\reac v$ is simply the number of nodes in the connected component of $v$, whereas for strongly-connected graphs is equal to $n$. We name such bound on the farness $\lfarnnb v{\reac v}$.

\section{The \ubcut\ Function} \label{sec:bfscut}

The \ubcut\ function is based on a simple idea: if the $k$-th biggest farness found until now is $x$, and if we are performing a BFS from vertex $v$ to compute its farness $\farn v$, we can stop as soon as we can guarantee that $\farn v \geq x$.

Informally, assume that we have already visited all nodes up to distance $d$: we can lower bound $S(v)=\sum_{w \in V} d(v,w)$ by setting distance $d+1$ to a number of vertices equal to the number of edges ``leaving'' level $d$, and distance $d+2$ to all the remaining reachable vertices. Then, this bound yields a lower bound on the farness of $v$. As soon as this lower bound is bigger than $x$, the \ubcut\ function may stop; if this condition never occurs, at the end of the BFS we have exactly computed the farness of $x$.

More formally, the following lemma defines a lower bound $\lsumcut dv{\reac v}$ on $S(v)$, which is computable after we have performed a BFS from $v$ up to level $d$, assuming we know the number $\reac v$ of vertices reachable from $v$ (this assumption is lifted in Sect.~\ref{sec:disc}).

\begin{lemma} \label{lem:boundokconn}
Given a graph $G=(V,E)$, a vertex $v \in V$, and an integer $d \geq 0$, let $\ball dv$ be the set of vertices at distance at most $d$ from $v$, $\sball dv=|\ball dv|$, and let $\uneig{d+1}v$ be an upper bound on the number of vertices at distance $d+1$ from $v$ (see Table~\ref{tbl:notation}). Then, 
\[S(v) \geq \lsumcut dv{\reac v}:=\sum_{w \in \ball dv}d(v,w)-\uneig{d+1}v+(d+2)(\reac v-\sball dv).\]
\end{lemma}
\begin{proof}
The sum of all the distances from $v$ is lower bounded by setting the correct distance to all vertices at distance at most $d$ from $v$, by setting distance $d+1$ to all vertices at distance $d+1$ (there are $\sneig{d+1}v$ such vertices), and by setting distance $d+2$ to all other vertices (there are $\reac v-\sball{d+1}v$ such vertices, where $\reac v$ is the number of vertices reachable from $v$ and $\sball{d+1}v$ is the number of vertices at distance at most $d+1$). More formally,
$\farn v \geq \sum_{w \in \ball dv}d(v,w) + (d+1)\sneig {d+1}v+(d+2)(\reac v-\sball {d+1}v).$

Since $\sball {d+1}v=\sneig {d+1}v+\sball dv$, we obtain that $\farn v \geq \sum_{w \in \ball dv}d(v,w) -\sneig {d+1}v+(d+2)(\reac v-\sball dv)$. We conclude because, by assumption, $\uneig {d+1}v$ is an upper bound on $\sneig {d+1}v$. 
\qed
\end{proof}
\begin{corollary}
For each vertex $v$ and for each $d \geq 0$, 
\[\farn v \geq \lfarncut dv{\reac v}:= \frac{(n-1)\lsumcut dv{\reac v}}{(\reac v -1)^2}.\]
\end{corollary}

It remains to define the upper bound $\uneig{d+1}v$: in the directed case, this bound is simply the sum of the out-degrees of vertices at distance $d$ from $v$. In the undirected case, since at least an edge from each vertex $v \in \neig dv$ is directed towards $\neig{d-1}v$, we may define $\uneig{d+1}v=\sum_{w \in \neig dv} \deg(w)-1$ (the only exception is $d=0$: in this case, $\uneig{1}v=\sneig 1v=\deg(v)$).

\begin{remark} \label{rem:impr}
When we are processing vertices at level $d$, if we process an edge $(x,y)$ where $y$ is already in the BFS tree, we can decrease $\uneig {d+1}v$ by one, obtaining a better bound.
\end{remark}

Assuming we know $\reac v$, all quantities necessary to compute $\lfarncut dv{\reac v}$ are available as soon as all vertices in $\ball dv$ are visited by a BFS. %Furthermore, if the graph is (strongly) connected, $\reac v=n$, while if the graph is undirected, $\reac v$ can be computed at the beginning in linear time. Hence, at least in these cases, we can define the procedure \ubcut$(v,x)$, which returns the farness $\farn v$ of vertex $v$ if $\farn v<x$, $x$ otherwise. 
This function performs a BFS starting from $v$, continuously updating the upper bound $\lfarncut dv{\reac v} \leq \farn v$ (the update is done whenever all nodes in $\Gamma_d(v)$ have been reached, or Remark~\ref{rem:impr} can be used). As soon as $\lfarncut dv{\reac v} \geq x$, we know that $\farn v \geq \lfarncut dv{\reac v} \geq x$, and we return $+\infty$.

Algorithm~\ref{alg:bfscut} is the pseudocode of the function \ubcut\ when implemented for directed graphs, assuming we know the number $\reac v$ of vertices reachable from each $v$ (for example, if the graph is strongly connected). This code can be easily adapted to all the other cases.

\newcommand{\secondalgo}{
\begin{small}
\begin{algorithm2e}[t]
\caption{The \ubcut$(v)$ function in the case of directed graphs, if $\reac v$ is known for each $v$.}
\label{alg:bfscut}
\SetKwFunction{bfscf}{\bfsc}
\SetKwFunction{preprocf}{\preproc}
\SetKwFunction{kth}{Kth}
\SetKwData{farn}{Farn}
\SetKwData{top}{Top}
$x \gets \farn(\top[k])$; // \farn and \top are global variables, as in Algorithm~\ref{alg:main}. \\
Create queue $Q$;\\
$Q$.enqueue($v$);\\
Mark $v$ as visited;\\
$d \gets 0$; $S \gets 0$;  $\tilde{\gamma} \gets \outdeg(v)$; $nd \gets 1$; \\
\While{$Q$ is not empty}{
    $u \gets Q$.dequeue(); \\
    \If{$d(v,u)>d$}{
        $d \gets d+1$;\\
        $\lfarncut dv{\reac v} \gets \frac{(n-1)\left(S-\tilde{\gamma}+(d+2)(\reac v-nd)\right)}{(\reac v-1)^2}$;\\ \label{line:ftilde}
        \lIf{$\lfarncut dv{\reac v} \geq x$}{\Return $+\infty$} \label{line:bfscutstop}
        $\tilde{\gamma} \gets 0$
    }
    \For{$w$ in adjacency list of $u$}{
        \If{$w$ is not visited}{
            $S \gets S+d(v,w)$;\\
            $\tilde{\gamma} \gets \tilde{\gamma}+\odegr w$;\\
            $nd \gets nd+1$;\\
            $Q$.enqueue($w$); \\
            Mark $w$ as visited
        } \Else { // we use Remark~\ref{rem:impr} \\ 
            $\lfarncut dv{\reac v} \gets \lfarncut dv{\reac v} + \frac{(n-1)}{(\reac v-1)^2}$; \\ \label{line:ftilde2}
            \lIf{$\lfarncut dv{\reac v} \geq x$}{\Return $x$} \label{line:bfscutstop2}
        }
    }
}
\Return $\frac{S(n-1)}{(\reac v -1)^2}$;
\end{algorithm2e}
\end{small}
}
\secondalgo

\section{The \ublb\ Function}
\label{sec:lblb}
\newcommand{\tsneiga}[2]{{\gamma}_{#1}}
\newcommand{\sneiga}[1]{{\Gamma}_{#1}}
\newcommand{\lev}[1]{\mathsf{L}(#1)}
Differently from \ubcut\ function, \ublb\ computes a complete BFS traversal, but uses information acquired during the traversal to update the bounds on the other nodes.
Let us first consider an undirected graph $G$ and let $s$ be the source node from which we are computing the BFS. We can see the distances $d(s,v)$ between $s$ and all the nodes $v$ reachable from $s$ as \textit{levels}: node $v$ is at level $i$ if and only if the distance between $s$ and $v$ is $i$, and we write $v \in \neig{i}{s}$ (or simply $v \in \sneiga{i}$ if $s$ is clear from the context). Let $i$ and $j$ be two levels, $i \leq j$. Then, the distance between any two nodes $v$ at level $i$ and $w$ at level $j$ must be at least $j - i$. Indeed, if $d(v,w)$ was smaller than $j-i$, $w$ would be at level $i+d(v,w) < j$, which contradicts our assumption. It follows directly that $\sum_{w\in V} | d(s,w) - d(s,v) |$ is a lower bound on $S(v)$, for all $v \in \Reac s$:
\begin{lemma}
$\sum_{w\in \Reac s} | d(s, w) - d(s, v) | \leq S(v) \quad \forall v \in \Reac s$.
\end{lemma}
To improve the approximation, we notice that the number of nodes at distance 1 from $v$ is exactly the degree of $v$. Therefore, all the other nodes $w$ such that $| d(s,v) - d(s,w)| \leq 1$ must be at least at distance 2 (with the only exception of $v$ itself, whose distance is of course 0). This way we can define the following lower bound on $S(v)$:
%
%\begin{equation}
%\label{eq:lbound1}
%\begin{split}
%\tilde{S}^{\mathrm{(un)}}_{\mathsf{L}}(v)  & \coloneqq  \degree (v) + 2( \#\{ w \in V : | l(w) - l(v) | \leq 1\} - \degree (v) - 1) + \sum\limits_{\substack{w \in V \\ | l(w) - l(v) | > 1}} | l(w) - l(v) | \\
%					& =  2\cdot \#\{ w \in V : | l(w) - l(v) | \leq 1\} + \left( \sum\limits_{\substack{w \in V \\ | l(w) - l(v) | > 1}} | l(w) - l(v) |  \right) - \degree (v) - 2
%\end{split}
%\end{equation}
\begin{equation*}
\begin{aligned}
%\label{eq:lbound1}
& 2( \#\{ w \in \Reac s : | d(s,w) - d(s, v) | \leq 1\} - \degree (v) - 1)   +\\
&+ \degree(v)+ \sum\limits_{\substack{w \in \Reac s \\ | d(s,w) - d(s,v) | > 1}} | d(s,w) - d(s,v) |,
\end{aligned}
\end{equation*}

\noindent that is:
\begin{equation}
\begin{aligned}
\label{eq:lbound1}
&2\cdot \sum_{|j - d(s,v)| \leq 1} \tsneiga{j}{}+ \sum_{|j - d(s,v)| >1} \tsneiga{j}{}\cdot |j - d(s,v)| - \degree (v) - 2,
\end{aligned}
\end{equation}
where $\tsneiga{j}{} = |\sneiga{j}{}|$.

Multiplying the bound of Eq.~(\ref{eq:lbound1}) by $\frac{(n-1)}{(r(v)-1)^2}$, we obtain a lower bound on the farness $\farn v$ of node $v$, named $\lfarnlb sv{\reac v}$.
A straightforward way to compute $\lfarnlb sv{\reac v}$ would be to first run the BFS from $s$ and then, for each node $v$, to consider the level difference between $v$ and all the other nodes. This would require $\O(n^2)$ operations, which is clearly too expensive. However, we can notice two things: First, the bounds of two nodes at the same level differ only by their degree. Therefore, for each level $i$, we can compute $2\cdot \sum_{|j - i| \leq 1} \tsneiga{j}{}+ \sum_{|j - i| >1} \tsneiga{j}{}\cdot |j - i| - 2$ only once and then subtract $\degree (v)$ for each node at level $i$. We call the quantity $2\cdot \sum_{|j - i| \leq 1} \tsneiga{j}{}+ \sum_{|j - i| >1} \tsneiga{j}{}\cdot |j - i| - 2$ the level-bound $\lev i$ of level $i$.
Second, we can prove that $\lev i$ can actually be written as a function of $\lev{i-1}$.
\begin{lemma}
\label{lemma:difference}
Let $\lev i := 2\cdot \sum_{|j - i| \leq 1} \tsneiga{j}{}+ \sum_{|j - i| >1} \tsneiga{j}{}\cdot |j - i| - 2$. Also, let $\tsneiga{j}{} = 0$ for $j\leq 0$ and $j>\mathsf{maxD}$, where $\mathsf{maxD} = \max_{v \in \Reac{s}} d(s,v)$.
Then $\lev i - \lev{i-1} = \sum_{j<i-2} \tsneiga{j}{s}-\sum_{j>i+1} \tsneiga{j}{s}$, $\forall i \in \{1, ..., \mathsf{maxD}\}$.
\end{lemma}
\begin{proof}
Since $\tsneiga{j}{} = 0$ for $j\leq 0$ and $j>\mathsf{maxD}$, we can write $\lev i$ as $2\cdot (\tsneiga{i-1}{}+\tsneiga{i}{}+\tsneiga{i+1}{})+ \sum_{|j - i| >1} \tsneiga{j}{}\cdot |j - i| - 2$, $\forall i \in \{1, ..., \mathsf{maxD}\}$.
The difference between $\lev i$ and $\lev{i-1}$ is: $2\cdot (\tsneiga{i-1}{s} + \tsneiga{i}{s} + \tsneiga{i+1}{s}) + \sum_{|j-i|>1} |j-i|\cdot \tsneiga{j}{s} - 2\cdot (\tsneiga{i-2}{s} + \tsneiga{i-1}{s} + \tsneiga{i}{s}) + \sum_{|j-i+1|>1} |j-i+1|\cdot \tsneiga{j}{s} = 2 \cdot (\tsneiga{i+1}{s} - \tsneiga{i-2}{s}) + 2\cdot \tsneiga{i-2}{s} - 2 \cdot \tsneiga{i+1}{s} + \sum_{j < i-2 \cup j >i+1} (|j-i|-|j-i+1|) \cdot \tsneiga{j}{s} = \sum_{j<i-2} \tsneiga{j}{s}-\sum_{j>i+1} \tsneiga{j}{s}$.
\qed
\end{proof}
\begin{algorithm2e}
 \begin{footnotesize}
\LinesNumbered
\SetKwFunction{BFSfrom}{BFSfrom}
\SetKwFunction{updateLowerBounds}{updateLowerBounds}
\SetKwFunction{enqueue}{enqueue}
\SetKwFunction{extractMin}{extractMin}
\SetKwFunction{seedNode}{seedNode}
\SetKwData{nL}{nL}
\SetKwData{Levels}{L}
\SetKwData{s}{sum}
\SetKwData{maxL}{maxD}
\SetKwData{sumD}{sum$\Gamma$}
 \SetKwInOut{Input}{Input}
 \SetKwInOut{Output}{Output}
\Input{A graph $G=(V,E)$, a source node $s$}
\Output{Lower bounds $\lfarnlb sv{\reac v}$ of each node $v\in \Reac s$}
$d \leftarrow \BFSfrom{s}$\; \label{lbound1:bfs}
$\maxL \leftarrow \max_{v \in V} d(s,v)$\;
$\sumD_{\leq 0} \leftarrow 0 $; $\sumD_{\leq -1} \leftarrow 0 $; $\sumD_{> {\maxL+1}} \leftarrow 0 $\;
\For{$i = 1, 2, ..., \maxL$} {
	$\sneiga{i} \leftarrow \{w \in V : d(s,w) = i \}$\;
	$\tsneiga{i}{} \leftarrow \# \sneiga{i}$\;
	$\sumD_{\leq i} \leftarrow \sumD_{\leq i-1} + \tsneiga{i}{}$\;
	$\sumD_{> i} \leftarrow |V| - \sumD_{\leq i}$\;
} \label{lbound1:end-init}
$\lev 1 \leftarrow \tsneiga{1}{} + \tsneiga{2}{} + \sumD_{> 2} -2$\; \label{lbound1:level1}
\For{$i = 2, ..., \maxL$} {
	$\lev i \leftarrow \lev {i-1} + \sumD_{\leq i-3} - \sumD_{> i+1}$\; \label{lbound1:leveli}
}
\For{$i = 1, ..., \maxL$} {
    \ForEach{$v \in \sneiga{i}$} {
        $\lfarnlb sv{\reac v} \leftarrow (\lev i - \degree(v)) \cdot \frac{(n-1)}{(r(v)-1)^2}$\; \label{lbound1:degree}
    }
}
\Return{$\lfarnlb sv{\reac v} \quad \forall v \in V$}
\end{footnotesize}
\caption{The \ublb\ function for undirected graphs}
\label{algo:lbound1}
\end{algorithm2e}

Algorithm~\ref{algo:lbound1} describes the computation of $\lfarnlb sv{\reac v}$. First, we compute all the distances between $s$ and the nodes in $\Reac s$ with a BFS, storing the number of nodes in each level and the number of nodes in levels $j\leq i$ and $j> i$ respectively (Lines~\ref{lbound1:bfs}~-~\ref{lbound1:end-init}). Then we compute the level bound $\lev 1$ of level $1$ according to its definition (Line~\ref{lbound1:level1}) and those of the other level according to Lemma~\ref{lemma:difference} (Line~\ref{lbound1:leveli}). The lower bound $\lfarnlb sv{\reac v}$ is then computed for each node $v$ by subtracting its degree to $\lev {d(s,v)}$ and normalizing (Line~\ref{lbound1:degree}). The complexity of Lines~\ref{lbound1:bfs}~-~\ref{lbound1:end-init} is that of running a BFS, \ie $\O(n+m)$. Line~\ref{lbound1:leveli} is repeated once for each level (which cannot be more than $n$) and Line~\ref{lbound1:degree} is repeated once for each node in $\Reac s$. Therefore, the following proposition holds.
\begin{proposition}
Computing the lower bound $\lfarnlb sv{\reac v}$ takes $\O(n + m)$ time.
\end{proposition}

For directed strongly-connected graphs, the result does not hold for nodes $w$ whose level is smaller than $l(v)$, since there might be a directed edge or a shortcut from $v$ to $w$. Yet, for nodes $w$ such that $d(s,w)>d(s,v)$, it is still true that $d(v,w)\geq d(s,w)-d(s,v)$. For the remaining nodes (apart from the outgoing neighbors of $v$),  we can only say that the distance must be at least 2. The upper bound $\lfarnlb sv{\reac v}$ for directed graphs can therefore be defined as:
\begin{equation}
\begin{aligned}
\label{eq:lbound1_dir}
&2\cdot \#\{ w \in \Reac s :  d(s,w) - d(s,v)  \leq 1\} \\
&+ \sum\limits_{\substack{w \in \Reac s \\  d(s,w) - d(s,v)  > 1}} ( d(s,w) - d(s,v) )  - \degree (v) - 2.
\end{aligned}
\end{equation}

The computation of $\lfarnlb sv{\reac v}$ for directed strongly-connected graphs is analogous to the one 
% for undirected 
described in Algorithm~\ref{algo:lbound1}.

\section{The Directed Disconnected Case} \label{sec:disc}

In the directed disconnected case, even if the time complexity of computing strongly connected components is linear in the input size, the time complexity of computing the number of reachable vertices is much bigger (assuming SETH, it cannot be $\O(m^{2-\epsilon})$ \cite{Borassi2016}). For this reason, when computing our upper bounds, we cannot rely on the exact value of $\reac v$: for now, let us assume that we know a lower bound $\alpha(v) \leq \reac v$ and an upper bound $\omega(v) \geq \reac v$. The definition of these bounds is postponed to Sect.~\ref{sec:alphaomega}.

Furthermore, let us assume that we have a lower bound $\lfarn v{\reac v}$ on the farness of $v$, depending on the number $\reac v$ of vertices reachable from $v$: in order to obtain a bound not depending on $\reac v$, the simplest approach is $\farn v \geq \lfarn v{\reac v} \geq \min_{\alpha(v) \leq r \leq \omega(v)} \lfarn vr$. However, during the algorithm, computing the minimum among all these values might be quite expensive, if $\omega(v)-\alpha(v)$ is big. In order to solve this issue, we find a small set $X\subseteq [\alpha(v),\omega(v)]$ such that $\min_{\alpha(v) \leq r \leq \omega(v)} \lfarn vr=\min_{r \in X} \lfarn vr$.

More specifically, we find a condition that is verified by ``many'' values of $r$, and that implies $\lfarn vr \geq \min\left(\lfarn v{r-1},\lfarn v{r+1}\right)$: this way, we may define $X$ as the set of values of $r$ that either do not verify this condition, or that are extremal points of the interval $[\alpha(v),\omega(v)]$ (indeed, all other values cannot be minima of $\lfarn vr$). Since all our bounds are of the form $\lfarn vr=\frac{(n-1)\lsum vr}{(r-1)^2}$, where $\lsum vr$ is a lower bound on $\sum_{w \in \Reac v}d(v,w)$, we state our condition in terms of the function $\lsum vr$. For instance, in the case of the \ubcut\ function, $\lsumcut dvr=\sum_{w \in \ball dv}d(v,w)-\uneig{d+1}v+(d+2)(r-\sball dv)$, as in Lemma~\ref{lem:boundokconn}.

\begin{lemma} \label{lem:disc}
Let $v$ be a vertex, and let $\lsum vr$ be a positive function such that $\lsum v{\reac v}) \leq \sum_{w \in \Reac v} d(v,w)$ (where $\reac v$ is the number of vertices reachable from $v$). Assume that $\lsum v{r+1}-\lsum vr \leq \lsum vr-\lsum v{r-1}$. Then, if $\lfarn vr:=\frac{(n-1)\lsum vr}{(r-1)^2}$ is the corresponding bound on the farness of $v$, $\min\left(\lfarn v{r+1}, \lfarn v{r-1}\right) \leq \lfarn vr$.
\end{lemma}
\begin{proof}
Let us define $d=\lsum v{r+1}-\lsum vr$. Then, $\lfarn v{r+1} \leq \lfarn vr$ if and only if $\frac{(n-1)\lsum v{r+1}}{r^2} \leq \frac{(n-1)\lsum v{r}}{(r-1)^2}$ if and only if $\frac{\lsum v{r}+d}{r^2} \leq \frac{\lsum v{r}}{(r-1)^2}$ if and only if $(r-1)^2(\lsum v{r}+d) \leq r^2\lsum v{r}$ if and only if $\lsum v{r}(r^2-(r-1)^2) \geq (r-1)^2d$ if and only if $\lsum v{r}(2r-1) \geq (r-1)^2d$.

Similarly, if $d'=\lsum v{r}-\lsum v{r-1}$, $\lfarn v{r-1} \leq \lfarn vr$ if and only if $\frac{(n-1)\lsum v{r-1}}{(r-2)^2} \leq \frac{(n-1)\lsum v{r}}{(r-1)^2}$ if and only if $\frac{\lsum v{r}-d'}{(r-2)^2} \leq \frac{\lsum v{r}}{(r-1)^2}$ if and only if $(r-1)^2(\lsum v{r}-d') \leq (r-2)^2\lsum v{r}$ if and only if $\lsum v{r}((r-1)^2-(r-2)^2) \leq (r-1)^2d'$ if and only if $\lsum v{r}(2r-3) \leq (r-1)^2d'$ if and only if $\lsum v{r}(2r-1) \leq (r-1)^2d'+2\lsum v{r}$.

We conclude that, assuming $d \leq d'$, $(r-1)^2d \leq (r-1)^2d' \leq (r-1)^2d+2\lsum v{r}$, and one of the two previous conditions is always satisfied.\qed
\end{proof}

\subsection{The Neighborhood-Based Lower Bound}

In the neighborhood-based lower bound, we computed upper bounds $\tsneig kv$ on $\neig kv$, and we defined the lower bound $\lsumnb v{r(v)} \leq \sum_{w \in \Reac v} d(v,w)$, by
\[
\lsumnb v{\reac v} := \sum_{k=1}^{\diam(G)} k \cdot \min \left \lbrace \tsneig{k}{v},\  \reac v - \sum_{i=0}^{k-1} \tsneig{i}{v},\ 0 \right \rbrace.
\]

The corresponding bound on $f(v)$ is $\lfarnnb v{\reac v}:=\frac{(n-1)\lsumnb v{\reac v}}{(\reac v-1)^2}$: let us apply Lemma~\ref{lem:disc} with $\lsum vr= \lsumnb vr$ and $\lfarn vr = \lfarnnb vr$. We obtain that the local minima of $\lfarnnb v{\reac v}$ are obtained on values $r$ such that $\lsumnb v{r+1}-\lsumnb v{r}>\lsumnb v{r}-\lsumnb v{r-1}$, that is, when $r=\sum_{i=0}^{l} \tsneig{i}{v}$ for some $l$. Hence, our final bound $\lfarnnbr v$ becomes:
\begin{small}
\begin{equation}\label{eq:nblbdisc}
\min\left(\lfarnnb v{\alpha(v)}, \lfarnnb v{\omega(v)}, \min \left\{\lfarnnb v{r}:\alpha(v) < r<\omega(v), r=\sum_{i=0}^{l} \tsneig{i}{v}\right\}\right).
\end{equation}
\end{small}
This bound can be computed with no overhead, by modifying Lines~\ref{line:if1}~-~\ref{line:else1-end} in Algorithm~\ref{algo:lbound2}. Indeed, when $\reac v$ is known, we have two cases: either $\texttt{nVisited[s]}<\reac v$, and we continue, or $\texttt{nVisited[s]} \geq \reac v$, and $\lsumnb v{\reac v}$ is computed. In the disconnected case, we need to distinguish three cases: 
\begin{itemize}
\item if $\texttt{nVisited[v]}<\alpha(v)$, we simply continue the computation;
\item if $\alpha(v)\leq \texttt{nVisited[v]}<\omega(v)$, we compute $\lfarnnb v{\texttt{nVisited[v]}}$, and we update the minimum in Eq.~\ref{eq:nblbdisc} (if this is the first occurrence of this situation, we also have to compute $\lfarnnb v{\alpha(v)}$);
\item if $\texttt{nVisited[v]}\geq \omega(v)$, we compute $\lfarnnb v{\omega(v)}$, and we update the minimum in Eq.~\ref{eq:nblbdisc}.
\end{itemize}

Since this procedure needs time $\O(1)$, it has no impact on the running time of the computation of the neighborhood-based lower bound.

\subsection{The \ubcut\ Function}

Let us apply Lemma~\ref{lem:disc} to the bound used in the \ubcut\ function. In this case, by Lemma~\ref{lem:boundokconn}, $\lsumcut dvr=\sum_{w \in \ball dv}d(v,w)-\uneig{d+1}v+(d+2)(r-\sball dv)$, and $\lsumcut dv{r+1}-\lsumcut dv{r}=d+2$, which does not depend on $r$. Hence, the condition in Lemma~\ref{lem:disc} is always verified, and the only values we have to analyze are $\alpha(v)$ and $\omega(v)$. Hence, the lower bound becomes $\farn v \geq \lfarncut dv{\reac v} \geq \min_{\alpha(v) \leq r \leq \omega(v)} \lfarncut dvr = \min(\lfarncut dv{\alpha(v)}, \lfarncut dv{\omega(v)})$ (which does not depend on $\reac v$).

This means that, in order to adapt the \ubcut\ function (Algorithm~\ref{alg:bfscut}), it is enough to replace Lines~\ref{line:ftilde},~\ref{line:ftilde2} in order to compute both $\lfarncut dv{\alpha(v)}$ and $\lfarncut dv{\omega(v)})$, and to replace Lines~\ref{line:bfscutstop},~\ref{line:bfscutstop2} in order to stop if $\min(\lfarncut dv{\alpha(v)},\lfarncut dv{\omega(v)}) \geq x$.

\subsection{The \ublb\ Function}

In this case, we do not apply Lemma~\ref{lem:disc} to obtain simpler bounds. Indeed, the \ublb\ function improves the bounds of vertices that are quite close to the source of the BFS, and hence are likely to be in the same component as this vertex. Consequently, if we perform a BFS from a vertex $s$, we can simply compute $\lfarnlb sv{\reac v}$ for all vertices in the same strongly connected component as $s$, and for these vertices we know the value $\reac v=\reac s$. The computation of better bounds for other vertices is left as an open problem.

\subsection{Computing $\alpha(v)$ and $\omega(v)$}\label{sec:alphaomega}

It now remains to compute $\alpha(v)$ and $\omega(v)$. This can be done during the preprocessing phase of our algorithm, in linear time. To this purpose, let us precisely define the node-weighted directed acyclic graph $\mathcal{G=(V,E)}$ of strongly connected components (in short, SCCs) corresponding to a directed graph $G=(V,E)$. In this graph, $\mathcal V$ is the set of SCCs of $G$, and, for any two SCCs $C,D \in \mathcal V$, $(C,D)\in\mathcal E$ if and only if there is an arc in $E$ from a node in $C$ to the a node in $D$. For each SCC $C\in \mathcal V$, the weight $\weight(C)$ of $C$ is equal to $|C|$, that is, the number of nodes in the SCC $C$. Note that the graph $\mathcal G$ is computable in linear time.

For each node $v \in C$, $\reac v=\sum_{D\in \Reac C}\weight(D)$, where $\Reac C$ denotes the set of SCCs that are reachable from $C$ in $\mathcal G$. This means that we simply need to compute a lower (respectively, upper) bound $\alpha_{SCC}(C)$ (respectively, $\omega_{SCC}(C)$) on $\sum_{D\in \mathcal R(C)}\weight(D)$, for every SCC $C$. To this aim, we first compute a topological sort $\{C_1, \dots, C_l\}$ of $\mathcal V$ (that is, if $(C_i,C_j) \in \mathcal E$, then $i<j$). Successively, we use a dynamic programming approach, and, by starting from $C_l$, we process the SCCs in reverse topological order, and we set:
\[
\alpha_{SCC}(C)=\weight(C) + \max_{(C, D) \in \mathcal E} \alpha_{SCC}(D) \quad\,\,
\omega_{SCC}(C)=\weight(C) + \sum_{(C, D) \in \mathcal E} \omega_{SCC}(D).
\]
Note that processing the SCCs in reverse topological ordering ensures that the values $\alpha(D)$ and $\omega(D)$ on the right hand side of these equalities are available when we process the SCC $C$. Clearly, the complexity of computing $\alpha(C)$ and $\omega(C)$, for each SCC $C$, is linear in the size of $\mathcal G$, which in turn is smaller than $G$.

Observe that the bounds obtained through this simple approach can be improved by using some ``tricks''. First of all, when the biggest SCC $\tilde{C}$ is processed, we do not use the dynamic programming approach and we exactly compute $\sum_{D\in \mathcal R(\tilde{C})}\weight(D)$ by performing a BFS starting from any node in $\tilde{C}$. This way, not only $\alpha(\tilde{C})$ and $\omega(\tilde{C})$ are exact, but also $\alpha_{SCC}(C)$ and $\omega_{SCC}(C)$ are improved for each SCC $C$ from which it is possible to reach $\tilde{C}$. Finally, in order to compute the upper bounds for the SCCs that are able to reach $\tilde{C}$, we can run the dynamic programming algorithm on the graph obtained from $\mathcal{G}$ by removing all components reachable from $\tilde{C}$, and we can then add $\sum_{D\in \mathcal R(\tilde{C})}\weight(D)$. 

The pseudocode is available in Algorithm~\ref{algo:alphaomega}.

\begin{algorithm2e}
 \begin{footnotesize}
\LinesNumbered
\SetKwFunction{BFSfrom}{BFSfrom}
\SetKwFunction{computeSCCGraph}{computeSCCGraph}
\SetKwData{SCCs}{SCCs}
\SetKwData{k}{k}
\SetKwData{nReachMaxSCC}{nReachMaxSCC}
 \SetKwInOut{Input}{Input}
 \SetKwInOut{Output}{Output}
\Input{A graph $G=(V,E)$}
\Output{Lower and upper bounds $\alpha(v),\omega(v)$ on the number of vertices reachable from $v$}
$(\mathcal{V},\mathcal{E},w) \gets \computeSCCGraph(G)$\;
$\tilde{C} \gets $ the biggest SCC\;
$\alpha_{SCC}(\tilde{C}), \omega_{SCC}(\tilde{C}) \gets $ the number of vertices reachable from $\tilde{C}$\;

\For{$X \in \mathcal V$ in reverse topological order} {
    \lIf{$X == \tilde{C}$}{continue}
    
    $\alpha_{SCC}(X), \omega_{SCC}(X), \omega'_{SCC}(X) \gets 0$
    \For{$Y$ neighbor of $X$ in $\mathcal{G}$} {
        $\alpha_{SCC}(X) \gets \max(\alpha_{SCC}(X), \alpha_{SCC}(Y))$\;
        $\omega_{SCC}(X) \gets \omega_{SCC}(X)+\omega_{SCC}(Y)$\;
        \lIf{$W$ not reachable from $\tilde{C}$}{$\omega'_{SCC}(X) \gets \omega'_{SCC}(X)+\omega_{SCC}(Y)$}
    }
    \lIf{$X$ reaches $\tilde{C}$}{$\omega_{SCC}(X) \gets \omega'_{SCC}(X)+\omega_{SCC}(\tilde{C})$}
    
    $\alpha_{SCC}(X) \gets \alpha_{SCC}(X)+\weight(X)$\;
    $\omega_{SCC}(X) \gets \omega_{SCC}(X)+\weight(X)$\;
}
\For{$v \in V$} {
    $\alpha(v)=\alpha_{SCC}($the component of $v)$\;
    $\omega(v)=\omega_{SCC}($the component of $v)$\;
}
\Return{$\alpha,\omega$}
\end{footnotesize}
\caption{Estimating the number of reachable vertices in directed disconnected graphs.}
\label{algo:alphaomega}
\end{algorithm2e}

\section{Experimental Results} \label{sec:exp}
In this section, we test the four variations of our algorithm on several real-world networks, in order to evaluate their performances. All the networks used in our experiments come from the datasets SNAP (\url{snap.stanford.edu/}), NEXUS (\url{nexus.igraph.org}), LASAGNE (\url{piluc.dsi.unifi.it/lasagne}), LAW (\url{law.di.unimi.it}), KONECT (\url{http://konect.uni-koblenz.de/networks/}, and IMDB (\url{www.imdb.com}). 
The platform for our tests is a shared-memory server with 256 GB RAM and 2x8 Intel(R) Xeon(R) E5-2680 cores (32 threads due to hyperthreading) at 2.7 GHz.
The algorithms are implemented in C++, building on the open-source \textit{NetworKit}
framework~\cite{Staudt2014}. 
%a server running an Intel(R) Xeon(R) CPU E5-4607 0, 2.20GHz, with 48 cores, 250GB RAM, running Ubuntu 14.04\_2 LTS; our code has been written in Java 1.7, and it is available at \codesite.

\subsection{Comparison with the State of the Art} \label{sec:expsmall}

In order to compare the performance of our algorithm with state-of-the-art approaches, we select 19 directed complex networks, 17 undirected complex networks, 6 directed road networks, and 6 undirected road networks (the undirected versions of the previous ones). The number of nodes of most of these networks ranges between $5\,000$ and $100\,000$. We test four different variations of our algorithm, that provide different implementations of the functions \cb\ and \ub\ (for more information, we refer to Sect.~\ref{sec:overview}):

\begin{description}
\item[\degcut] uses the conservative strategies \cbdeg\ and \ubcut;
\item[\degbound] uses the conservative strategy \cbdeg\ and the aggressive strategy \ublb; 
\item[\nbcut] uses the aggressive strategy \cbnb\ and the conservative strategy \ubcut;
\item[\nbbound] uses the aggressive strategies \cbnb\ and \ublb.
\end{description}

We compare these algorithms with our implementations of the best existing algorithms for top-$k$ closeness centrality.\footnote{Note that the source code of our competitors is not available.} The first one \cite{Olsen2014} is based on a pruning technique and on $\Delta$-BFS, a method to reuse information collected during a BFS from a node to speed up a BFS from one of its in-neighbors; we denote this algorithm as \deltabfs. The second one, \sampl, provides top-$k$ closeness centralities with high probability \cite{Okamoto2008}. It performs some BFSes from a random sample of nodes to estimate the closeness centrality of all the other nodes, then it computes the exact centrality of all the nodes whose estimate is big enough. Note that this algorithm requires the input graph to be (strongly) connected: for this reason, differently from the other algorithms, we have run this algorithm on the largest (strongly) connected component of the input graph. Furthermore, this algorithm offers different tradeoffs between the time needed by the sampling phase and the second phase: in our tests, we try all possible tradeoffs, and we choose the best alternative in each input graph (hence, our results are upper bounds on the real performance of the \sampl\ algorithm). 

In order to perform a fair comparison, we consider the \emph{improvement factor}, which is defined as $\frac{mn}{\mvis}$ in directed graphs, $\frac{2mn}{\mvis}$ in undirected graphs, where $\mvis$ is the number of arcs visited during the algorithm, and $mn$ (resp., $2mn$) is an estimate of the number of arcs visited by the \emph{textbook} algorithm in directed (resp., undirected) graphs (this estimate is correct whenever the graph is connected). Note that the improvement factor does not depend on the implementation, nor on the machine used for the algorithm, and it does not consider parts of the code that need subquadratic time in the worst case. These parts are negligible in our algorithm, because their worst case running time is $\O(n\log n)$ or $\O(mD)$ where $D$ is the diameter of the graph, but they can be significant when considering the competitors. For instance, in the particular case of \deltabfs, we have just counted the arcs visited in BFS and $\Delta$-BFS, ignoring all the operations done in the pruning phases (see \cite{Olsen2014}). 

We consider the geometric mean of the improvement factors over all graphs in the dataset. In our opinion, this quantity is more informative than the arithmetic mean, which is highly influenced by the maximum value: for instance, in a dataset of 20 networks, if all improvement factors are $1$ apart from one, which is $10\,000$, the arithmetic mean is more than $500$, which makes little sense, while the geometric mean is about $1.58$. Our choice is further confirmed by the geometric standard deviation, which is always quite small. 

The results are summarised in Table~\ref{tab:comparisoncomp} for complex networks and Table~\ref{tab:comparisonstreet} for street networks. For the improvement factors of each graph, we refer to Appendix~\ref{app:compext}.

\begin{table}[t]
\centering
\begin{scriptsize}
\tbl{Complex networks: geometric mean and standard deviation of the improvement factors of the algorithm in \cite{Olsen2014} (\deltabfs), the algorithm in \cite{Okamoto2008} (\sampl), and the four variations of the new algorithm (\degcut, \degbound, \nbcut, \nbbound). \label{tab:comparisoncomp}}{
\begin{tabular}{|@{\hspace{3pt}}c@{\hspace{3pt}}|l|r|r|r|r|r|r|}
\hline
 & &  \multicolumn{2}{c|}{\textsc{Directed}} & \multicolumn{2}{c|}{\textsc{Undirected}} & \multicolumn{2}{c|}{\textsc{Both}}\\
$k$ & \textsc{Algorithm} & \multicolumn{1}{c|}{\textsc{GMean}}& \multicolumn{1}{c|}{\textsc{GStdDev}} & \multicolumn{1}{c|}{\textsc{GMean}}& \multicolumn{1}{c|}{\textsc{GStdDev}}& \multicolumn{1}{c|}{\textsc{GMean}}& \multicolumn{1}{c|}{\textsc{GStdDev}}\\
\hline
1 & \deltabfs & 21.24 & 5.68 & 11.11 & 2.91 & 15.64 & 4.46 \\ 
 & \sampl & 1.71 & 1.54 & 2.71 & 1.50 & 2.12 & 1.61 \\ 
 & \degcut & 104.20 & 6.36 & 171.77 & 6.17 & 131.94 & 6.38 \\ 
 & \degbound & 3.61 & 3.50 & 5.83 & 8.09 & 4.53 & 5.57 \\ 
 &\nbcut & \textbf{123.46} & 7.94 & \textbf{257.81} & 8.54 & \textbf{174.79} & 8.49 \\ 
 & \nbbound & 17.95 & 10.73 & 56.16 & 9.39 & 30.76 & 10.81 \\
 \hline
10 & \deltabfs & 21.06 & 5.65 & 11.11 & 2.90 & 15.57 & 4.44 \\ 
 & \sampl & 1.31 & 1.31 & 1.47 & 1.11 & 1.38 & 1.24 \\ 
 & \degcut & 56.47 & 5.10 & 60.25 & 4.88 & 58.22 & 5.00 \\ 
 & \degbound & 2.87 & 3.45 & 2.04 & 1.45 & 2.44 & 2.59 \\ 
 &\nbcut & \textbf{58.81} & 5.65 & \textbf{62.93} & 5.01 & \textbf{60.72} & 5.34 \\ 
 & \nbbound & 9.28 & 6.29 & 10.95 & 3.76 & 10.03 & 5.05 \\ 
 \hline
100 & \deltabfs & 20.94 & 5.63 & 11.11 & 2.90 & 15.52 & 4.43 \\ 
 & \sampl & 1.30 & 1.31 & 1.46 & 1.11 & 1.37 & 1.24 \\ 
 & \degcut & 22.88 & 4.70 & 15.13 & 3.74 & 18.82 & 4.30 \\ 
 & \degbound & 2.56 & 3.44 & 1.67 & 1.36 & 2.09 & 2.57 \\ 
&\nbcut & \textbf{23.93} & 4.83 & \textbf{15.98} & 3.89 & \textbf{19.78} & 4.44 \\ 
 & \nbbound & 4.87 & 4.01 & 4.18 & 2.46 & 4.53 & 3.28 \\
 \hline
\end{tabular}}
\end{scriptsize}
\end{table}

\begin{table}[t]
\centering
\begin{scriptsize}
\tbl{Street networks: geometric mean and standard deviation of the improvement factors of the algorithm in \cite{Olsen2014} (\deltabfs), the algorithm in \cite{Okamoto2008} (\sampl), and the four variations of the new algorithm (\degcut, \degbound, \nbcut, \nbbound).\label{tab:comparisonstreet}}{
\begin{tabular}{|@{\hspace{3pt}}c@{\hspace{3pt}}|l|r|r|r|r|r|r|}
\hline
 & &  \multicolumn{2}{c|}{\textsc{Directed}} & \multicolumn{2}{c|}{\textsc{Undirected}} & \multicolumn{2}{c|}{\textsc{Both}}\\
$k$ & \textsc{Algorithm} & \multicolumn{1}{c|}{\textsc{GMean}}& \multicolumn{1}{c|}{\textsc{GStdDev}} & \multicolumn{1}{c|}{\textsc{GMean}}& \multicolumn{1}{c|}{\textsc{GStdDev}}& \multicolumn{1}{c|}{\textsc{GMean}}& \multicolumn{1}{c|}{\textsc{GStdDev}}\\
\hline
1 & \deltabfs & 4.11 & 1.83 & 4.36 & 2.18 & 4.23 & 2.01 \\ 
 & \sampl & 3.39 & 1.28 & 3.23 & 1.28 & 3.31 & 1.28 \\ 
 & \degcut & 4.14 & 2.07 & 4.06 & 2.06 & 4.10 & 2.07 \\ 
 & \degbound & 187.10 & 1.65 & 272.22 & 1.67 & 225.69 & 1.72 \\ 
 & \nbcut & 4.12 & 2.07 & 4.00 & 2.07 & 4.06 & 2.07 \\ 
&\nbbound & \textbf{250.66} & 1.71 & \textbf{382.47} & 1.63 & \textbf{309.63} & 1.74 \\ 
 \hline
10  & \deltabfs & 4.04 & 1.83 & 4.28 & 2.18 & 4.16 & 2.01 \\ 
 & \sampl & 2.93 & 1.24 & 2.81 & 1.24 & 2.87 & 1.24 \\ 
 & \degcut & 4.09 & 2.07 & 4.01 & 2.06 & 4.05 & 2.07 \\ 
 & \degbound & 172.06 & 1.65 & 245.96 & 1.68 & 205.72 & 1.72 \\ 
 & \nbcut & 4.08 & 2.07 & 3.96 & 2.07 & 4.02 & 2.07 \\ 
&\nbbound & \textbf{225.26} & 1.71 & \textbf{336.47} & 1.68 & \textbf{275.31} & 1.76 \\ 
 \hline
 100 & \deltabfs & 4.03 & 1.82 & 4.27 & 2.18 & 4.15 & 2.01 \\ 
 & \sampl & 2.90 & 1.24 & 2.79 & 1.24 & 2.85 & 1.24 \\ 
 & \degcut & 3.91 & 2.07 & 3.84 & 2.07 & 3.87 & 2.07 \\ 
 & \degbound & 123.91 & 1.56 & 164.65 & 1.67 & 142.84 & 1.65 \\ 
 & \nbcut & 3.92 & 2.08 & 3.80 & 2.09 & 3.86 & 2.08 \\ 
&\nbbound & \textbf{149.02} & 1.59 & \textbf{201.42} & 1.69 & \textbf{173.25} & 1.67\\
 \hline
\end{tabular}}
\end{scriptsize}
\end{table}

On complex networks, the best algorithm is \nbcut: when $k=1$, the improvement factors are always bigger than $100$, up to $258$, when $k=10$ they are close to $60$, and when $k=100$ they are close to $20$. Another good option is \degcut, which achieves results similar to \nbcut, but it has almost no overhead at the beginning (while \nbcut\ needs a preprocessing phase with cost $\O(mD)$). Furthermore, \degcut\ is very easy to implement, becoming a very good candidate for state-of-the-art graph libraries. The improvement factors of the competitors are smaller: \deltabfs\ has improvement factors between $10$ and $20$, and \sampl\ provides almost no improvement with respect to the \emph{textbook} algorithm.

We also test our algorithm on the three complex unweighted networks analysed in \cite{Olsen2014}, respectively called \texttt{web-Google} (\texttt{Web} in \cite{Olsen2014}), \texttt{wiki-Talk} (\texttt{Wiki} in \cite{Olsen2014}), and \texttt{com-dblp} (\texttt{DBLP} in \cite{Olsen2014}). In the \texttt{com-dblp} graph (resp. \texttt{web-Google}), our algorithm \nbcut\ computed the top 10 nodes in about $17$ seconds (resp., less than $2$ minutes) on the whole graph, having $1\,305\,444$ nodes (resp., $875\,713$), while \deltabfs\ needed about $25$ minutes (resp. $4$ hours) on a subgraph of $400\,000$ nodes. In the graph \texttt{wiki-Talk}, \nbcut\ needed $8$ seconds for the whole graph having $2\,394\,385$ nodes, instead of about $15$ minutes on a subgraph with 1 million nodes. These results are available in Table~\ref{tab:compbigapp} in the Appendix.

On street networks, the best option is \nbbound: for $k=1$, the average improvement is about $250$ in the directed case and about $382$ in the undirected case, and it always remains bigger than $150$, even for $k=100$. It is worth noting that also the performance of \degbound\ are quite good, being at least $70\%$ of \nbbound. Even in this case, the \degbound\ algorithm offers some advantages: it is very easy to be implemented, and there is no overhead in the first part of the computation. All the competitors perform relatively poorly on street networks, since their improvement is always smaller than $5$.

Overall, we conclude that the preprocessing function \cbnb\ always leads to better results (in terms of visited edges) than \cbdeg, but the difference is quite small: hence, in some cases, \cbdeg\ could be even preferred, because of its simplicity. Conversely, the performance of \ubcut\ is very different from the performance of \ublb: the former works much better on complex networks, while the latter works much better on street networks. Currently, these two approaches exclude each other: an open problem left by this work is the design of a ``combination'' of the two, that works both in complex networks and in street networks. Finally, the experiments show that the best variation of our algorithm outperforms all competitors in all frameworks considered: both in complex and in street networks, both in directed and undirected graphs.

\paragraph{Harmonic Centrality}
As mentioned in the introduction, all our methods can be easily generalized to any centrality measure in the form $c(v) = \sum_{w \neq v} f(d(v, w))$, where $f$ is a decreasing function such that $f(+\infty) = 0$. We also implemented a version of \degcut, \degbound, \nbcut\ and \nbbound\ for \textit{harmonic centrality}, which is defined as $h(v) = \sum_{w \neq v} \frac{1}{d(v, w)}$.
Also for harmonic centrality, we compute the improvement factors on the textbook algorithm.

For the complex networks used in our experiments, finding the $k$ nodes with highest harmonic centrality is \textit{always faster} than finding the $k$ nodes with highest closeness, for all four methods and $k$ values in $\{ 1, 10, 100\}$. For example, for \nbcut\ and $k=1$, the geometric mean\footnote{We report the geometric mean over both directed and undirected networks.} of the improvement factors is 486.07, whereas for closeness it is 174.79 (as reported in Table~\ref{tab:comparisoncomp}). 

For street networks, the version of harmonic centrality is faster than the version for closeness for \degcut\ and \nbcut, but it is slower for \degbound\ and \nbbound. In particular, the average (geometric mean) improvement factor of \nbbound\ for harmonic centrality is 103.58 for $k=1$, 93.49 for $k=10$ and 62.22 for $k=100$, which is about a factor 3 smaller than the improvement factor of  \nbbound\ for closeness (see Table~\ref{tab:comparisonstreet}). Nevertheless, this is significantly faster than the textbook algorithm.

\subsection{Real-World Large Networks} \label{sec:expbig}

In this section, we run our algorithm on bigger inputs, by considering a dataset containing $23$ directed networks, $15$ undirected networks, and $5$ road networks, with up to $3\,774\,768$ nodes and $117\,185\,083$ edges. On this dataset, we run the fastest variant of our algorithm (\degbound\ in complex networks, \nbbound\ in street networks), using $64$ threads (however, the server used has only $16$ cores and runs $32$ threads with hyperthreading; we account for memory latency in
graph computations by oversubscribing slightly).

Once again, we consider the \emph{improvement factor}, which is defined as $\frac{mn}{\mvis}$ in directed graphs, $\frac{2mn}{\mvis}$ in undirected graphs. It is worth observing that we are able to compute for the first time the $k$ most central nodes of networks with millions of nodes and hundreds of millions of arcs, with $k=1$, $k=10$, and $k=100$. The detailed results are shown in Table~\ref{tab:compbigapp} in the Appendix, where for each network we report the running time and the improvement factor. A summary of these results is available in Table~\ref{tab:compbig}, which contains the geometric means of the improvement factors, with the corresponding standard deviations. 

\begin{table}[t]
\centering
\begin{scriptsize}
\tbl{Big networks: geometric mean and standard deviation of the improvement factors of the best variation of the new algorithm (\degbound\ in complex networks, \nbbound\ in street networks).\label{tab:compbig}}{
\begin{tabular}{|l|r|r|r|r|r|r|r|}
\hline
 & &  \multicolumn{2}{c|}{\textsc{Directed}} & \multicolumn{2}{c|}{\textsc{Undirected}} & \multicolumn{2}{c|}{\textsc{Both}}\\
Input & \multicolumn{1}{c|}{$k$}  &  \multicolumn{1}{c|}{\textsc{GMean}}& \multicolumn{1}{c|}{\textsc{GStdDev}} & \multicolumn{1}{c|}{\textsc{GMean}}& \multicolumn{1}{c|}{\textsc{GStdDev}}& \multicolumn{1}{c|}{\textsc{GMean}}& \multicolumn{1}{c|}{\textsc{GStdDev}}\\
\hline
 & 1 & 742.42 & 2.60 & 1681.93 & 2.88 & 1117.46 & 2.97 \\ 
Street & 10 & 724.72 & 2.67 & 1673.41 & 2.92 & 1101.25 & 3.03 \\ 
 & 100 & 686.32 & 2.76 & 1566.72 & 3.04 & 1036.95 & 3.13 \\ 
 \hline
 & 1 & 247.65 & 11.92 & 551.51 & 10.68 & 339.70 & 11.78 \\ 
Complex & 10 & 117.45 & 9.72 & 115.30 & 4.87 & 116.59 & 7.62 \\
 & 100 & 59.96 & 8.13 & 49.01 & 2.93 & 55.37 & 5.86 \\
 \hline
\end{tabular}}
\end{scriptsize}
\end{table}

For $k=1$, the geometric mean of the improvement factors is always above $200$ in complex networks, and above $700$ in street networks. In undirected graphs, the improvement factors are even bigger: close to $500$ in complex networks and close to $1\,600$ in street networks. For bigger values of $k$, the performance does not decrease significantly: on complex networks, the improvement factors are bigger than or very close to $50$, even for $k=100$. In street networks, the performance loss is even smaller, always below $10\%$ for $k=100$.

Regarding the robustness of the algorithm, we outline that the algorithm always achieves performance improvements bigger than $\sqrt{n}$ in street networks, and that in complex networks, with $k=1$, $64\%$ of the networks have improvement factors above $100$, and $33\%$ of the networks above $1\,000$. In some cases, the improvement factor is even bigger: in the \texttt{com-Orkut} network, our algorithm for $k=1$ is almost $35\,000$ times faster than the \emph{textbook} algorithm.

In our experiments, we also report the running time of our algorithm. Even for $k=100$, a few minutes are sufficient to conclude the computation on most networks, and, in all but two cases, the total time is smaller than $3$ hours. For $k=1$, the computation always terminates in at most $1$ hour and a half, apart from two street networks where it needs less than $2$ hours and a half. Overall, the total time needed to compute the most central vertex in all the networks is smaller than $1$ day. In contrast to this, if we extrapolate the results in Tables~\ref{tab:comparisoncomp} and~\ref{tab:comparisonstreet}, it seems plausible that the fastest competitor OLH would require a month or so.

\section{IMDB Case Study} \label{sec:imdb}

In this section, we apply the new algorithm \nbbound\ to analyze the IMDB graph, where nodes are actors, and two actors are connected if they played together in a movie (TV-series are ignored). The data collected comes from the website \url{http://www.imdb.com}: in line with \url{http://oracleofbacon.org}, we decide to exclude some genres from our database: awards-shows, documentaries, game-shows, news, realities and talk-shows. We analyse snapshots of the actor graph, taken every 5 years from 1940 to 2010, and 2014. The results are reported in Table~\ref{tab:imdb} and Table~\ref{tab:imdb2} in the Appendix.

\paragraph{The Algorithm}
Thanks to this experiment, we can evaluate the performance of our algorithm on increasing snapshots of the same graph. This way, we can have an informal idea on the asymptotic behavior of its complexity. In Figure~\ref{fig:actorperformance}, we have plotted the improvement factor with respect to the number of nodes: if the improvement factor is $I$, the running time is $\O(\frac{mn}{I})$. Hence, assuming that $I=cn$ for some constant $c$ (which is approximately verified in the actor graph, as shown by Figure~\ref{fig:actorperformance}), the running time is linear in the input size.
The total time needed to perform the computation on all snapshots is little more than $30$ minutes for $k=1$, and little more than $45$ minutes for $k=10$. 

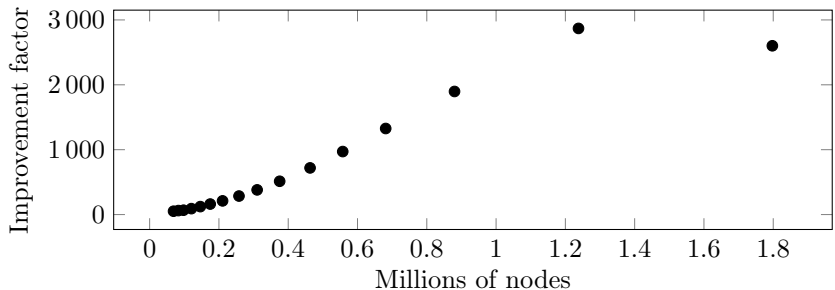
\begin{figure}
\centering
\begin{tikzpicture}
\tikzset{
did/.style={rectangle,draw=none,fill=none,inner sep=0cm},
every node/.style={did}}

\pgfplotsset{cycle list={{mark=*}},height=4.5cm,width=0.8\textwidth}
\begin{axis}[xlabel=Millions of nodes,
			ylabel=Improvement factor,
			y label style={did,at={(0,0.5)}},
			x tick label style={did, minimum height = 0.5cm},
			x label style={did},
			y tick label style={did, inner sep=3pt},
			legend pos={north west}
			]
\pgfkeys{/pgf/number format/.cd, set thousands separator={\,}}

\addplot+[mark=*,only marks] coordinates {
(.069011,51.742991662)
(.083068,61.4642531784)
(.097824,67.5031941787)
(.120430,91.4572612002)
(.146253,122.6341427434)
(.174826,162.0572633208)
(.210527,211.053381443)
(.257896,285.5700800513)
(.310278,380.5246528171)
(.375322,513.3969825689)
(.463078,719.2069886806)
(.557373,971.1081551636)
(.681358,1326.5304767698)
(.880032,1897.313214871)
(1.237879,2869.1380172861)
(1.797446,2601.5196086576)
};
%   \addlegendentry{$k=1$};
%\addplot+[mark=none,dashed] coordinates {
%(.069011,32.95)
%(.083068,40.73)
%(.097824,44.72)
%(.120430,61.52)
%(.146253,80.50)
%(.174826,111.51)
%(.210527,159.32)
%(.257896,221.07)
%(.310278,296.24)
%(.375322,416.27)
%(.463078,546.77)
%(.557373,694.72)
%(.681358,838.53)
%(.880032,991.89)
%(1.237879,976.63)
%(1.797446,1390.32)
%};
  % \addlegendentry{$k=10$};
\end{axis}
\end{tikzpicture}
\caption{Growth of performance ratio with respect to the number of nodes ($k=1$).}
\label{fig:actorperformance}
\end{figure}

\paragraph{The Results} In 2014, the most central actor is Michael Madsen, whose career spans 25 years and more than 170 films. Among his most famous appearances, he played as \emph{Jimmy Lennox} in \emph{Thelma \& Louise} (Ridley Scott, 1991), as \emph{Glen Greenwood} in \emph{Free Willy} (Simon Wincer, 1993), as \emph{Bob} in \emph{Sin City} (Frank Miller, Robert Rodriguez, Quentin Tarantino), and as \emph{Deadly Viper Budd} in \emph{Kill Bill} (Quentin Tarantino, 2003-2004).
The second is Danny Trejo, whose most famous movies are  %Heat (Michael Mann, 1995), Desperado (Robert Rodriguez, 1995), and Con Air (Simon West, 1997), 
\emph{Heat} (Michael Mann, 1995), where he played as \emph{Trejo}, \emph{Machete} (Ethan Maniquis, Robert Rodriguez, 2010) and \emph{Machete Kills} (Robert Rodriguez, 2013), where he played as \emph{Machete}.
The third ``actor'' is not really an actor: he is the German dictator Adolf Hitler: he was also the most central actor in 2005 and 2010, and he was in the top 10 since 1990. This a consequence of his appearances in several archive footages, that were re-used in several movies (he counts 775 credits, even if most of them are in documentaries or TV shows, which were eliminated). Among the movies where Adolf Hitler is credited, we find \emph{Zelig} (Woody Allen, 1983), and \emph{The Imitation Game} (Morten Tyldum, 2014). Among the other most central actors, we find many people who played a lot of movies, and most of them are quite important actors. However, this ranking does not discriminate between important roles and marginal roles: for instance, the actress Bess Flowers is not widely known, because she rarely played significant roles, but she appeared in over 700 movies in her 41 years career, and for this reason she was the most central for 30 years, between 1950 and 1980. Finally, it is worth noting that we never find Kevin Bacon in the top 10, even if he became famous for the ``Six Degrees of Kevin Bacon'' game (\url{http://oracleofbacon.org}). In this game the player receives an actor $x$ and has to find a path of length at most $6$ from $x$ to Kevin Bacon in the actor graph. Kevin Bacon was chosen as the goal because he played in several movies, and he was thought to be one of the most central actors: this work shows that, actually, he is quite far from the top. Indeed, his closeness centrality is $0.336$, while the most central actor has centrality $0.354$, the 10th actor has centrality $0.350$, and the 100th actor 
%(Rip Torn) 
has centrality $0.341$.

\section{Wikipedia Case Study} \label{sec:wiki}

In this section, we apply the new algorithm \nbbound\ to analyze the Wikipedia graph, where nodes are pages, and there is a directed edge from page $p$ to page $q$ if $p$ contains a link to $q$. The data collected comes from DBPedia 3.7 (\url{http://wiki.dbpedia.org/}). We analyse both the standard graph and the reverse graph, which contains an edge from page $p$ to page $q$ if $q$ contains a link to $p$. The $10$ most central pages are available in Table~\ref{tab:wiki}.

\begin{table}[t]
\centering
\tbl{Top $10$ pages in Wikipedia directed graph, both in the standard graph and in the reversed graph.\label{tab:wiki}}{
\begin{tabular}{|r|l|l|}
\hline
Position & Standard Graph & Reversed Graph \\
\hline
\textsc{1st} & 1989 & United States\\
\textsc{2nd} & 1967 & World War II\\
\textsc{3rd} & 1979 & United Kingdom\\
\textsc{4th} & 1990 & France\\
\textsc{5th} & 1970 & Germany\\
\textsc{6th} & 1991 & English language\\
\textsc{7th} & 1971 & Association football\\
\textsc{8th} & 1976 & China\\
\textsc{9th} & 1945 & World War I\\
\textsc{10th} & 1965 & Latin\\
\hline
\end{tabular}}
\end{table}

\paragraph{The Algorithm}
In the standard graph, the improvement factor is $1\,784$ for $k=1$, $1\,509$ for $k=10$, and $870$ for $k=100$. The total running time is about $39$ minutes for $k=1$, $45$ minutes for $k=10$, and less than $1$ hour and $20$ minutes for $k=100$. In the reversed graph, the algorithm performs even better: the improvement factor is $87\,918$ for $k=1$, $71\,923$ for $k=10$, and $21\,989$ for $k=100$. The total running times are less than $3$ minutes for both $k=1$ and $k=10$, and less than $10$ minutes for $k=100$.

\paragraph{The Results}
If we consider the standard graph, the results are quite unexpected: indeed, all the most central pages are years (the first is \emph{1989}). However, this is less surprising if we consider that these pages contain a lot of links to events that happened in that year: for instance, the out-degree of \emph{1989} is $1\,560$, and the links contain pages from very different topics: historical events, like the fall of Berlin wall, days of the year, different countries where particular events happened, and so on. A similar argument also works for other years: indeed, the second page is \emph{1967} (with out-degree $1\,438$), and the third is \emph{1979} (with out-degree $1\,452$). Furthermore, all the 10 most central pages have out-degree at least $1\,269$. Overall, we conclude that the central page in the Wikipedia standard graph are not the ``intuitively important'' pages, but they are the pages that have a biggest number of links to pages with different topics, and this maximum is achieved by pages related to years.

Conversely, if we consider the reversed graph, the most central page is \emph{United States}, confirming a common conjecture. Indeed, in \url{http://wikirank.di.unimi.it/}, it is shown that the United States are the center according to harmonic centrality, and many other measures (however, in that work, the ranking is only approximated). A further evidence for this conjecture comes from the Six Degree of Wikipedia game (\url{http://thewikigame.com/6-degrees-of-wikipedia}), where a player is asked to go from one page to the other following the smallest possible number of link: a hard variant of this game forces the player not to pass the \emph{United States} page, which is considered to be central. In this work, we show that this conjecture is true. The second page is \emph{World War II}, and the third is \emph{United Kingdom}, in line with the results obtained by other centrality measures (see \url{http://wikirank.di.unimi.it/}), especially for the first two pages.

Overall, we conclude that most of the central pages in the reversed graph are nations, and that the results capture our intuitive notion of ``important'' pages in Wikipedia. Thanks to this new algorithm, we can compute these pages in a bit more than 1 hour for the original graph, and less than 10 minutes for the reversed one.

\section{Conclusions}
In this paper we have presented a hardness result on the computation of the most central vertex in a graph, according to closeness centrality. Then, we have presented a very simple algorithm for the exact computation of the $k$ most central vertices. Even if the time complexity of the new algorithm is equal to the time complexity of the textbook algorithm (which, in any case, cannot be improved in general), we have shown that in practice the former improves the latter by several orders of magnitude. 
We have also shown that the new algorithm outperforms the state of the art (whose time complexity is still equal to the complexity of the textbook algorithm), and we have computed for the first time the most central nodes in networks with millions of nodes and hundreds of millions of edges. Finally, we have considered as a case study several snapshots of the IMDB actor network, and the Wikipedia graph.

\section*{Acknowledgments}
This work is partially supported by German Research Foundation (DFG) grant ME-3619/3-1 (FINCA) within the Priority Programme 1736 \textit{Algorithms for Big Data} and by by the Italian Ministry of Education, University, and Research (MIUR) under PRIN 2012C4E3KT national research project AMANDA -- \textit{Algorithmics for MAssive and Networked DAta}.

\bibliographystyle{ACM-Reference-Format-Journals}
\bibliography{library}

\clearpage
\appendix

\section*{Appendix}

\section{Comparison with the State of the Art: Detailed Results} \label{app:compext}
\begin{table}[ht]
\centering
\begin{scriptsize}
\setlength{\tabcolsep}{2.7pt}
\tbl{Detailed comparison of the improvement factors, with $k=1$. \label{tab:compext}}{
\begin{tabular}{|l|r|r|r|r|r|r|}
\multicolumn{7}{c}{\textbf{Directed Street}} \\
\hline
Network & \deltabfs & \sampl & \degcut & \degbound & \nbcut & \nbbound \\ 
\hline
 faroe-islands & 4.080 & 3.742 & 4.125 & 338.011 & 4.086 & 437.986 \\ 
 liechtenstein & 2.318 & 2.075 & 2.114 & 130.575 & 2.115 & 137.087 \\ 
 isle-of-man & 2.623 & 3.740 & 2.781 & 224.566 & 2.769 & 314.856 \\ 
 malta & 5.332 & 4.351 & 4.147 & 73.836 & 4.141 & 110.665 \\ 
 belize & 2.691 & 3.969 & 2.606 & 253.866 & 2.595 & 444.849 \\ 
 azores & 13.559 & 3.038 & 19.183 & 230.939 & 19.164 & 266.488 \\
\hline
\multicolumn{7}{c}{\textbf{Undirected Street}} \\
\hline
Network & \deltabfs & \sampl & \degcut & \degbound & \nbcut & \nbbound \\ 
\hline
faroe-islands & 4.126 & 3.276 & 4.118 & 361.593 & 3.918 & 444.243 \\ 
 liechtenstein & 2.318 & 2.027 & 2.107 & 171.252 & 2.122 & 183.240 \\ 
 isle-of-man & 2.613 & 3.661 & 2.767 & 266.734 & 2.676 & 370.194 \\ 
 malta & 4.770 & 4.164 & 3.977 & 122.729 & 3.958 & 232.622 \\ 
 belize & 2.565 & 3.945 & 2.510 & 340.270 & 2.481 & 613.778 \\ 
 azores & 22.406 & 2.824 & 18.654 & 589.985 & 18.810 & 727.528 \\
\hline
\multicolumn{7}{c}{\textbf{Directed Complex}} \\
\hline
Network & \deltabfs & \sampl & \degcut & \degbound & \nbcut & \nbbound \\ 
\hline
polblogs & 3.201 & 1.131 & 31.776 & 1.852 & 31.974 & 5.165 \\ 
 out.opsahl-openflights & 13.739 & 1.431 & 73.190 & 2.660 & 73.888 & 18.255 \\ 
 ca-GrQc & 9.863 & 1.792 & 36.673 & 3.630 & 38.544 & 6.307 \\ 
 out.subelj\_jung-j\_jung-j & 125.219 & 1.203 & 79.559 & 1.024 & 79.882 & 1.897 \\ 
 p2p-Gnutella08 & 5.696 & 1.121 & 66.011 & 4.583 & 81.731 & 6.849 \\ 
 out.subelj\_jdk\_jdk & 116.601 & 1.167 & 74.300 & 1.023 & 74.527 & 1.740 \\ 
 wiki-Vote & 9.817 & 2.760 & 261.242 & 1.479 & 749.428 & 395.278 \\ 
 p2p-Gnutella09 & 5.534 & 1.135 & 41.214 & 4.650 & 43.236 & 6.101 \\ 
 ca-HepTh & 7.772 & 2.121 & 40.068 & 3.349 & 42.988 & 5.217 \\ 
 freeassoc & 33.616 & 1.099 & 12.638 & 2.237 & 12.700 & 2.199 \\ 
 ca-HepPh & 7.682 & 2.836 & 10.497 & 3.331 & 10.516 & 4.387 \\ 
 out.lasagne-spanishbook & 13.065 & 2.553 & 1871.296 & 7.598 & 6786.506 & 3160.750 \\ 
 out.cfinder-google & 16.725 & 1.782 & 38.321 & 2.665 & 25.856 & 3.020 \\ 
 ca-CondMat & 7.382 & 3.526 & 409.772 & 5.448 & 517.836 & 29.282 \\ 
 out.subelj\_cora\_cora & 14.118 & 1.700 & 14.098 & 1.345 & 14.226 & 2.299 \\ 
 out.ego-twitter & 2824.713 & 1.000 & 1870.601 & 28.995 & 3269.183 & 278.214 \\ 
 out.ego-gplus & 722.024 & 1.020 & 3481.943 & 236.280 & 3381.029 & 875.111 \\ 
 as-caida20071105 & 20.974 & 3.211 & 2615.115 & 1.737 & 2837.853 & 802.273 \\ 
 cit-HepTh & 4.294 & 3.045 & 16.259 & 1.514 & 16.398 & 3.290 \\ 
\hline
\multicolumn{7}{c}{\textbf{Undirected Complex}} \\
\hline
Network & \deltabfs & \sampl & \degcut & \degbound & \nbcut & \nbbound \\ 
\hline
HC-BIOGRID & 5.528 & 1.581 & 15.954 & 3.821 & 14.908 & 3.925 \\ 
 facebook\_combined & 10.456 & 3.726 & 56.284 & 18.786 & 56.517 & 98.512 \\ 
 Mus\_musculus & 18.246 & 1.743 & 70.301 & 3.253 & 104.008 & 7.935 \\ 
 Caenorhabditis\_elegans & 11.446 & 2.258 & 86.577 & 2.140 & 110.677 & 9.171 \\ 
 ca-GrQc & 6.567 & 1.904 & 38.279 & 3.551 & 41.046 & 6.824 \\ 
 as20000102 & 19.185 & 2.402 & 1550.351 & 3.213 & 1925.916 & 498.000 \\ 
 advogato & 8.520 & 2.018 & 315.024 & 18.181 & 323.163 & 142.654 \\ 
 p2p-Gnutella09 & 3.744 & 2.336 & 90.252 & 1.708 & 100.427 & 13.846 \\ 
 hprd\_pp & 6.543 & 2.397 & 392.853 & 2.091 & 407.261 & 63.953 \\ 
 ca-HepTh & 7.655 & 2.075 & 42.267 & 3.308 & 46.326 & 5.593 \\ 
 Drosophila\_melanogaster & 5.573 & 2.346 & 69.457 & 1.822 & 75.456 & 6.904 \\ 
 oregon1\_010526 & 20.474 & 3.723 & 1603.739 & 2.703 & 1798.822 & 399.071 \\ 
 oregon2\_010526 & 17.330 & 4.748 & 1138.475 & 2.646 & 1227.105 & 520.955 \\ 
 Homo\_sapiens & 6.689 & 2.700 & 1475.113 & 1.898 & 1696.909 & 130.381 \\ 
 GoogleNw & 15.591 & 8.389 & 107.902 & 15763.000 & 15763.000 & 15763.000 \\ 
 dip20090126\_MAX & 2.883 & 3.826 & 5.833 & 6.590 & 5.708 & 7.392 \\ 
 com-amazon.all.cmty & 415.286 & 2.499 & 5471.982 & 3.297 & 8224.693 & 373.294 \\ 
 \hline
\end{tabular}}
\end{scriptsize}
\end{table}
\begin{table}[ht]
\setlength{\tabcolsep}{3pt}
\centering
\begin{scriptsize}
\tbl{Detailed comparison of the improvement factors, with $k=10$. \label{tab:compext10}}{
\begin{tabular}{|l|r|r|r|r|r|r|}
\multicolumn{7}{c}{\textbf{Directed Street}} \\
\hline
Network & \deltabfs & \sampl & \degcut & \degbound & \nbcut & \nbbound \\ 
\hline
 faroe-islands & 3.713 & 2.884 & 4.037 & 290.626 & 4.025 & 361.593 \\ 
 liechtenstein & 2.318 & 2.002 & 2.104 & 111.959 & 2.106 & 116.713 \\ 
 isle-of-man & 2.623 & 2.933 & 2.711 & 209.904 & 2.720 & 288.123 \\ 
 malta & 5.325 & 3.861 & 4.094 & 70.037 & 4.086 & 101.546 \\ 
 belize & 2.690 & 3.638 & 2.592 & 244.275 & 2.580 & 416.210 \\ 
 azores & 13.436 & 2.644 & 19.043 & 222.073 & 19.045 & 254.206 \\
\hline
\multicolumn{7}{c}{\textbf{Undirected Street}} \\
\hline
Network & \deltabfs & \sampl & \degcut & \degbound & \nbcut & \nbbound \\ 
\hline
 faroe-islands & 3.702 & 2.594 & 4.046 & 320.588 & 3.848 & 388.713 \\ 
 liechtenstein & 2.316 & 1.965 & 2.097 & 142.047 & 2.114 & 150.608 \\ 
 isle-of-man & 2.612 & 2.889 & 2.695 & 241.431 & 2.636 & 323.185 \\ 
 malta & 4.768 & 3.615 & 3.920 & 115.574 & 3.910 & 208.192 \\ 
 belize & 2.564 & 3.634 & 2.496 & 323.257 & 2.469 & 563.820 \\ 
 azores & 22.392 & 2.559 & 18.541 & 539.032 & 18.712 & 653.372 \\
\hline
\multicolumn{7}{c}{\textbf{Directed Complex}} \\
\hline
Network & \deltabfs & \sampl & \degcut & \degbound & \nbcut & \nbbound \\ 
\hline
 polblogs & 3.199 & 1.039 & 13.518 & 1.496 & 13.544 & 2.928 \\ 
 out.opsahl-openflights & 13.739 & 1.130 & 32.297 & 1.984 & 32.405 & 6.867 \\ 
 ca-GrQc & 9.863 & 1.356 & 25.238 & 3.096 & 25.786 & 4.565 \\ 
 out.subelj\_jung-j\_jung-j & 124.575 & 1.000 & 79.284 & 1.024 & 79.657 & 1.884 \\ 
 p2p-Gnutella08 & 5.684 & 1.064 & 12.670 & 3.241 & 12.763 & 3.599 \\ 
 out.subelj\_jdk\_jdk & 116.228 & 1.000 & 74.106 & 1.023 & 74.363 & 1.730 \\ 
 wiki-Vote & 9.812 & 1.205 & 166.941 & 1.453 & 174.775 & 25.411 \\ 
 p2p-Gnutella09 & 5.532 & 1.084 & 16.293 & 3.624 & 16.265 & 4.213 \\ 
 ca-HepTh & 7.772 & 1.586 & 31.314 & 3.013 & 32.604 & 4.356 \\ 
 freeassoc & 33.414 & 1.034 & 10.612 & 2.210 & 10.704 & 2.178 \\ 
 ca-HepPh & 7.682 & 2.077 & 10.322 & 3.042 & 10.340 & 4.010 \\ 
 out.lasagne-spanishbook & 13.063 & 1.483 & 303.044 & 1.067 & 351.262 & 94.351 \\ 
 out.cfinder-google & 16.725 & 1.413 & 36.364 & 2.665 & 24.765 & 3.017 \\ 
 ca-CondMat & 7.382 & 2.318 & 91.209 & 3.507 & 93.548 & 7.027 \\ 
 out.subelj\_cora\_cora & 13.699 & 1.287 & 12.763 & 1.334 & 12.909 & 2.072 \\ 
 out.ego-twitter & 2689.884 & 1.000 & 1817.032 & 28.157 & 2872.213 & 218.411 \\ 
 out.ego-gplus & 722.024 & 1.000 & 951.983 & 201.949 & 1085.361 & 482.204 \\ 
 as-caida20071105 & 20.974 & 1.615 & 997.996 & 1.371 & 1266.443 & 448.729 \\ 
 cit-HepTh & 4.030 & 2.179 & 11.361 & 1.486 & 11.423 & 2.832 \\ 
\hline
\multicolumn{7}{c}{\textbf{Undirected Complex}} \\
\hline
Network & \deltabfs & \sampl & \degcut & \degbound & \nbcut & \nbbound \\ 
\hline
 HC-BIOGRID & 5.528 & 1.240 & 10.714 & 3.102 & 10.036 & 3.058 \\ 
 facebook\_combined & 10.456 & 1.292 & 9.103 & 2.236 & 9.371 & 2.694 \\ 
 Mus\_musculus & 18.246 & 1.316 & 18.630 & 2.279 & 20.720 & 3.288 \\ 
 Caenorhabditis\_elegans & 11.445 & 1.405 & 58.729 & 1.904 & 68.905 & 7.605 \\ 
 ca-GrQc & 6.567 & 1.340 & 26.050 & 3.052 & 26.769 & 5.011 \\ 
 as20000102 & 19.185 & 1.529 & 196.538 & 1.314 & 209.674 & 52.210 \\ 
 advogato & 8.520 & 1.405 & 131.173 & 2.043 & 132.207 & 11.155 \\ 
 p2p-Gnutella09 & 3.744 & 1.632 & 79.093 & 1.623 & 87.357 & 12.941 \\ 
 hprd\_pp & 6.543 & 1.436 & 47.945 & 1.837 & 47.866 & 8.620 \\ 
 ca-HepTh & 7.655 & 1.546 & 32.612 & 2.961 & 34.407 & 4.677 \\ 
 Drosophila\_melanogaster & 5.573 & 1.672 & 50.840 & 1.646 & 54.637 & 5.743 \\ 
 oregon1\_010526 & 20.474 & 1.451 & 418.099 & 1.282 & 429.161 & 109.549 \\ 
 oregon2\_010526 & 17.330 & 1.560 & 364.277 & 1.302 & 371.929 & 71.186 \\ 
 Homo\_sapiens & 6.689 & 1.599 & 81.496 & 1.620 & 82.250 & 15.228 \\ 
 GoogleNw & 15.591 & 1.320 & 23.486 & 1.252 & 23.053 & 2.420 \\ 
 dip20090126\_MAX & 2.881 & 1.836 & 4.055 & 4.556 & 4.065 & 4.498 \\ 
 com-amazon.all.cmty & 414.765 & 1.618 & 3407.016 & 3.279 & 3952.370 & 199.386 \\ 
 \hline
\end{tabular}}
\end{scriptsize}
\end{table}
\begin{table}[ht]
\centering
\setlength{\tabcolsep}{2.7pt}
\begin{scriptsize}
\tbl{Detailed comparison of the improvement factors, with $k=100$. \label{tab:compext100}}{
\begin{tabular}{|l|r|r|r|r|r|r|}
\multicolumn{7}{c}{\textbf{Directed Street}} \\
\hline
Network & \deltabfs & \sampl & \degcut & \degbound & \nbcut & \nbbound \\ 
\hline
 faroe-islands & 3.713 & 2.823 & 3.694 & 150.956 & 3.691 & 168.092 \\ 
 liechtenstein & 2.318 & 1.998 & 2.078 & 84.184 & 2.086 & 86.028 \\ 
 isle-of-man & 2.620 & 2.902 & 2.551 & 139.139 & 2.567 & 167.808 \\ 
 malta & 5.282 & 3.850 & 3.933 & 56.921 & 3.942 & 76.372 \\ 
 belize & 2.688 & 3.617 & 2.526 & 184.718 & 2.516 & 268.634 \\ 
 azores & 13.334 & 2.628 & 18.380 & 194.724 & 18.605 & 220.013 \\ 
\hline
\multicolumn{7}{c}{\textbf{Undirected Street}} \\
\hline
Network & \deltabfs & \sampl & \degcut & \degbound & \nbcut & \nbbound \\ 
\hline
 faroe-islands & 3.702 & 2.548 & 3.693 & 159.472 & 3.523 & 171.807 \\ 
 liechtenstein & 2.311 & 1.959 & 2.072 & 96.782 & 2.095 & 99.768 \\ 
 isle-of-man & 2.607 & 2.847 & 2.533 & 153.859 & 2.468 & 183.982 \\ 
 malta & 4.758 & 3.605 & 3.745 & 89.929 & 3.730 & 137.538 \\ 
 belize & 2.562 & 3.629 & 2.428 & 226.582 & 2.406 & 323.257 \\ 
 azores & 22.345 & 2.548 & 18.092 & 411.760 & 18.384 & 476.253 \\
\hline
\multicolumn{7}{c}{\textbf{Directed Complex}} \\
\hline
Network & \deltabfs & \sampl & \degcut & \degbound & \nbcut & \nbbound \\ 
\hline
 polblogs & 3.198 & 1.037 & 3.951 & 1.245 & 3.961 & 1.731 \\ 
 out.opsahl-openflights & 13.739 & 1.124 & 5.524 & 1.456 & 5.553 & 1.740 \\ 
 ca-GrQc & 9.863 & 1.339 & 11.147 & 2.353 & 10.407 & 2.926 \\ 
 out.subelj\_jung-j\_jung-j & 123.393 & 1.000 & 78.473 & 1.021 & 78.798 & 1.787 \\ 
 p2p-Gnutella08 & 5.684 & 1.063 & 6.611 & 2.935 & 7.750 & 3.278 \\ 
 out.subelj\_jdk\_jdk & 114.210 & 1.000 & 73.522 & 1.020 & 73.755 & 1.669 \\ 
 wiki-Vote & 9.812 & 1.186 & 61.375 & 1.236 & 60.475 & 9.436 \\ 
 p2p-Gnutella09 & 5.531 & 1.083 & 6.370 & 3.109 & 7.650 & 3.508 \\ 
 ca-HepTh & 7.772 & 1.570 & 16.135 & 2.477 & 16.747 & 3.135 \\ 
 freeassoc & 33.266 & 1.032 & 6.314 & 2.154 & 6.428 & 2.138 \\ 
 ca-HepPh & 7.682 & 2.032 & 9.605 & 2.549 & 9.619 & 3.340 \\ 
 out.lasagne-spanishbook & 13.063 & 1.467 & 56.689 & 1.043 & 80.069 & 33.271 \\ 
 out.cfinder-google & 16.725 & 1.392 & 13.521 & 2.655 & 12.298 & 2.722 \\ 
 ca-CondMat & 7.382 & 2.288 & 16.884 & 2.602 & 16.950 & 2.824 \\ 
 out.subelj\_cora\_cora & 13.231 & 1.280 & 11.171 & 1.315 & 11.350 & 1.870 \\ 
 out.ego-twitter & 2621.659 & 1.000 & 1574.836 & 26.893 & 1908.731 & 110.236 \\ 
 out.ego-gplus & 722.024 & 1.000 & 522.333 & 181.754 & 522.576 & 236.280 \\ 
 as-caida20071105 & 20.974 & 1.606 & 17.971 & 1.216 & 18.694 & 5.479 \\ 
 cit-HepTh & 3.969 & 2.143 & 8.867 & 1.466 & 9.068 & 2.662 \\ 
\hline
\multicolumn{7}{c}{\textbf{Undirected Complex}} \\
\hline
Network & \deltabfs & \sampl & \degcut & \degbound & \nbcut & \nbbound \\ 
\hline
HC-BIOGRID & 5.528 & 1.236 & 4.452 & 2.154 & 4.345 & 1.999 \\ 
 facebook\_combined & 10.456 & 1.292 & 3.083 & 1.470 & 3.074 & 1.472 \\ 
 Mus\_musculus & 18.245 & 1.305 & 7.940 & 1.944 & 9.518 & 2.631 \\ 
 Caenorhabditis\_elegans & 11.445 & 1.391 & 11.643 & 1.463 & 12.296 & 3.766 \\ 
 ca-GrQc & 6.567 & 1.331 & 11.311 & 2.346 & 10.389 & 3.105 \\ 
 as20000102 & 19.185 & 1.512 & 7.318 & 1.174 & 7.956 & 3.593 \\ 
 advogato & 8.520 & 1.398 & 32.629 & 1.706 & 33.166 & 7.784 \\ 
 p2p-Gnutella09 & 3.744 & 1.625 & 11.378 & 1.374 & 11.867 & 3.695 \\ 
 hprd\_pp & 6.543 & 1.422 & 21.053 & 1.547 & 22.191 & 3.468 \\ 
 ca-HepTh & 7.655 & 1.539 & 16.406 & 2.454 & 17.030 & 3.301 \\ 
 Drosophila\_melanogaster & 5.573 & 1.655 & 29.115 & 1.487 & 30.979 & 4.614 \\ 
 oregon1\_010526 & 20.474 & 1.443 & 13.300 & 1.163 & 14.611 & 6.569 \\ 
 oregon2\_010526 & 17.330 & 1.530 & 18.203 & 1.173 & 21.758 & 7.258 \\ 
 Homo\_sapiens & 6.689 & 1.577 & 19.350 & 1.445 & 20.182 & 3.080 \\ 
 GoogleNw & 15.591 & 1.320 & 16.224 & 1.172 & 16.506 & 2.010 \\ 
 dip20090126\_MAX & 2.880 & 1.815 & 2.789 & 2.602 & 2.784 & 2.546 \\ 
 com-amazon.all.cmty & 414.765 & 1.605 & 1368.675 & 3.236 & 1654.150 & 97.735 \\
 \hline
\end{tabular}}
\end{scriptsize}
\end{table}

\mbox{}

\clearpage

\section{Real-World Large Networks Experiments: Detailed Results}
\vspace{-1em}
\begin{table}[!h]
\centering
\setlength{\tabcolsep}{2.5pt}
\begin{scriptsize}
\tbl{Detailed comparison of the improvement factors on big networks. \label{tab:compbigapp}}{
\begin{tabular}{|l|r|r|r|r|r|r|r|r|}
\multicolumn{9}{c}{\textbf{Directed Street}} \\ 
\hline
 &  &  & \multicolumn{2}{c|}{$k=1$} & \multicolumn{2}{c|}{$k=10$} & \multicolumn{2}{c|}{$k=100$} \\ 
Input & Nodes & Edges & Impr. & Time & Impr. & Time & Impr. & Time \\ 
\hline
egypt & 1054242 & 2123036 & 144.91 & 0:03:55 & 132.86 & 0:04:25 & 116.74 & 0:04:48 \\ 
new\_zealand & 2759124 & 5562944 & 447.55 & 0:02:34 & 443.95 & 0:02:35 & 427.31 & 0:02:38 \\ 
india & 16230072 & 33355834 & 1370.32 & 0:43:42 & 1369.05 & 0:44:17 & 1326.31 & 0:45:05 \\ 
california & 16905319 & 34303746 & 1273.66 & 0:54:56 & 1258.12 & 0:56:00 & 1225.73 & 0:56:02 \\ 
north\_am & 35236615 & 70979433 & 1992.68 & 2:25:58 & 1967.87 & 2:29:25 & 1877.78 & 2:37:14 \\ 
\hline
\multicolumn{9}{c}{\textbf{Undirected Street}} \\ 
\hline
 &  &  & \multicolumn{2}{c|}{$k=1$} & \multicolumn{2}{c|}{$k=10$} & \multicolumn{2}{c|}{$k=100$} \\ 
Input & Nodes & Edges & Impr. & Time & Impr. & Time & Impr. & Time \\ 
\hline
egypt & 1054242 & 1159808 & 344.86 & 0:01:54 & 340.30 & 0:01:54 & 291.71 & 0:02:11 \\ 
new\_zealand & 2759124 & 2822257 & 811.75 & 0:02:47 & 786.52 & 0:03:02 & 734.20 & 0:03:02 \\ 
india & 16230072 & 17004400 & 2455.38 & 0:44:21 & 2484.70 & 0:44:38 & 2422.40 & 0:44:21 \\ 
california & 16905319 & 17600566 & 2648.08 & 0:39:15 & 2620.17 & 0:42:04 & 2504.86 & 0:44:19 \\ 
north\_am & 35236615 & 36611653 & 7394.88 & 1:13:37 & 7530.80 & 1:15:01 & 7263.78 & 1:10:28 \\ 
\hline
\multicolumn{9}{c}{\textbf{Directed Complex}} \\ 
\hline
 &  &  & \multicolumn{2}{c|}{$k=1$} & \multicolumn{2}{c|}{$k=10$} & \multicolumn{2}{c|}{$k=100$} \\ 
Input & Nodes & Edges & Impr. & Time & Impr. & Time & Impr. & Time \\ 
\hline
cit-HepTh & 27769 & 352768 & 16.34 & 0:00:01 & 11.41 & 0:00:01 & 9.06 & 0:00:02 \\ 
cit-HepPh & 34546 & 421534 & 23.68 & 0:00:01 & 19.88 & 0:00:01 & 14.41 & 0:00:02 \\ 
p2p-Gnut31 & 62586 & 147892 & 194.19 & 0:00:01 & 44.24 & 0:00:01 & 19.34 & 0:00:04 \\ 
soc-Eps1 & 75879 & 508837 & 243.14 & 0:00:01 & 43.75 & 0:00:01 & 33.60 & 0:00:05 \\ 
soc-Slash0811 & 77360 & 828161 & 1007.70 & 0:00:00 & 187.46 & 0:00:00 & 21.09 & 0:00:18 \\ 
twitter\_comb & 81306 & 2684592 & 1024.32 & 0:00:01 & 692.96 & 0:00:01 & 145.68 & 0:00:05 \\ 
Slash090221 & 82140 & 549202 & 177.82 & 0:00:02 & 162.30 & 0:00:02 & 108.53 & 0:00:03 \\ 
gplus\_comb & 107614 & 24476570 & 1500.35 & 0:00:04 & 235.17 & 0:00:04 & 62.54 & 0:02:19 \\ 
soc-sign-eps & 131828 & 840799 & 225.91 & 0:00:03 & 161.58 & 0:00:03 & 39.26 & 0:00:16 \\ 
email-EuAll & 265009 & 418956 & 4724.80 & 0:00:00 & 3699.48 & 0:00:00 & 1320.22 & 0:00:01 \\ 
web-Stanford & 281903 & 2312497 & 13.59 & 0:04:00 & 8.70 & 0:04:00 & 7.47 & 0:07:15 \\ 
web-NotreD & 325729 & 1469679 & 1690.08 & 0:00:02 & 132.83 & 0:00:02 & 66.88 & 0:00:49 \\ 
amazon0601 & 403394 & 3387388 & 10.81 & 0:14:54 & 8.87 & 0:14:54 & 6.84 & 0:22:04 \\ 
web-BerkStan & 685230 & 7600595 & 3.95 & 1:36:21 & 3.67 & 1:36:21 & 3.47 & 1:49:12 \\ 
web-Google & 875713 & 5105039 & 228.61 & 0:01:51 & 96.63 & 0:01:51 & 38.69 & 0:10:29 \\ 
youtube-links & 1138494 & 4942297 & 662.78 & 0:01:33 & 200.68 & 0:01:33 & 125.72 & 0:07:02 \\ 
in-2004 & 1382870 & 16539643 & 43.68 & 0:41:45 & 29.89 & 0:41:45 & 16.68 & 1:48:42 \\ 
trec-wt10g & 1601787 & 8063026 & 33.86 & 0:36:01 & 20.39 & 0:36:01 & 16.73 & 1:10:54 \\ 
soc-pokec & 1632803 & 22301964 & 21956.64 & 0:00:17 & 2580.43 & 0:06:14 & 1106.90 & 0:12:35 \\ 
zhishi-hudong & 1984484 & 14682258 & 30.37 & 1:25:38 & 27.71 & 1:25:38 & 24.95 & 1:53:27 \\ 
zhishi-baidu & 2141300 & 17632190 & 44.05 & 1:17:52 & 38.61 & 1:17:52 & 23.17 & 3:08:05 \\ 
wiki-Talk & 2394385 & 5021410 & 34863.42 & 0:00:08 & 28905.76 & 0:00:08 & 9887.18 & 0:00:18 \\ 
cit-Patents & 3774768 & 16518947 & 9454.04 & 0:02:07 & 8756.77 & 0:02:07 & 8340.18 & 0:02:13 \\ 
\hline
\multicolumn{9}{c}{\textbf{Undirected Complex}} \\ 
\hline
 &  &  & \multicolumn{2}{c|}{$k=1$} & \multicolumn{2}{c|}{$k=10$} & \multicolumn{2}{c|}{$k=100$} \\ 
Input & Nodes & Edges & Impr. & Time & Impr. & Time & Impr. & Time \\ 
\hline
ca-HepPh & 12008 & 118489 & 10.37 & 0:00:00 & 10.20 & 0:00:00 & 9.57 & 0:00:01 \\ 
CA-AstroPh & 18772 & 198050 & 62.47 & 0:00:00 & 28.87 & 0:00:01 & 14.54 & 0:00:01 \\ 
CA-CondMat & 23133 & 93439 & 247.35 & 0:00:00 & 84.48 & 0:00:00 & 17.06 & 0:00:01 \\ 
email-Enron & 36692 & 183831 & 365.92 & 0:00:00 & 269.80 & 0:00:00 & 41.95 & 0:00:01 \\ 
loc-brightkite & 58228 & 214078 & 308.03 & 0:00:00 & 93.85 & 0:00:01 & 53.49 & 0:00:02 \\ 
flickrEdges & 105938 & 2316948 & 39.61 & 0:00:23 & 17.89 & 0:00:55 & 15.39 & 0:01:16 \\ 
gowalla & 196591 & 950327 & 2412.26 & 0:00:01 & 33.40 & 0:01:18 & 28.13 & 0:01:33 \\ 
com-dblp & 317080 & 1049866 & 500.83 & 0:00:10 & 300.61 & 0:00:17 & 99.64 & 0:00:52 \\ 
com-amazon & 334863 & 925872 & 37.76 & 0:02:21 & 31.33 & 0:02:43 & 18.68 & 0:04:34 \\ 
com-lj.all & 477998 & 530872 & 849.57 & 0:00:07 & 430.72 & 0:00:13 & 135.14 & 0:00:45 \\ 
com-youtube & 1134890 & 2987624 & 2025.32 & 0:00:32 & 167.45 & 0:06:44 & 110.39 & 0:09:16 \\ 
soc-pokec & 1632803 & 30622564 & 46725.71 & 0:00:18 & 8664.33 & 0:02:16 & 581.52 & 0:18:12 \\ 
as-skitter & 1696415 & 11095298 & 185.91 & 0:19:06 & 164.24 & 0:21:53 & 132.38 & 0:27:06 \\ 
com-orkut & 3072441 & 117185083 & 23736.30 & 0:02:32 & 255.17 & 2:54:58 & 69.23 & 15:02:06 \\ 
youtube-u-g & 3223585 & 9375374 & 11473.14 & 0:01:07 & 91.17 & 2:07:23 & 66.23 & 2:54:12 \\
\hline
\end{tabular}}
\end{scriptsize}
\end{table}
\FloatBarrier
\newpage

\section{IMDB Case Study: Detailed Results}
\begin{table*}[!h]
\begin{scriptsize}
\centering
\setlength{\tabcolsep}{1.2pt}
\tbl{Detailed ranking of the IMDB actor graph. \label{tab:imdb}}{
\begin{tabular}{|l|c|c|c|c|}
\hline
 & \textbf{1940} & \textbf{1945} & \textbf{1950} & \textbf{1955}\\ 
%\textsc{Nodes} & 69\,011 & 83\,068 & 97\,824 & 120\,430\\ 
%\textsc{Edges} & 3\,417\,144 & 5\,160\,584 & 6\,793\,184 & 8\,674\,159\\ 
%\textsc{Impr ($k=1$)} & 51.74 & 61.46 & 67.50 & 91.46 \\ 
%\textsc{Impr ($k=10$)} & 32.95 & 40.73 & 44.72 & 61.52 \\ 
\hline
\textsc{1} &           Semels, Harry (I)  &               Corrado, Gino  &               Flowers, Bess  &               Flowers, Bess \\ 
\textsc{2} &               Corrado, Gino  &               Steers, Larry  &               Steers, Larry  &            Harris, Sam (II) \\ 
\textsc{3} &               Steers, Larry  &               Flowers, Bess  &               Corrado, Gino  &               Steers, Larry \\ 
\textsc{4} &              Bracey, Sidney  &           Semels, Harry (I)  &            Harris, Sam (II)  &               Corrado, Gino \\ 
\textsc{5} &              Lucas, Wilfred  &              White, Leo (I)  &           Semels, Harry (I)  &          Miller, Harold (I) \\ 
\textsc{6} &              White, Leo (I)  &            Mortimer, Edmund  &           Davis, George (I)  &            Farnum, Franklyn \\ 
\textsc{7} &           Martell, Alphonse  &               Boteler, Wade  &             Magrill, George  &             Magrill, George \\ 
\textsc{8} &           Conti, Albert (I)  &             Phelps, Lee (I)  &             Phelps, Lee (I)  &               Conaty, James \\ 
\textsc{9} &               Flowers, Bess  &                 Ring, Cyril  &                 Ring, Cyril  &           Davis, George (I) \\ 
\textsc{10} &                 Sedan, Rolfe  &               Bracey, Sidney  &              Moorhouse, Bert  &               Cording, Harry \\ 
\hline
%\multicolumn{5}{c}{\vspace{-.5\baselineskip}} \\ 
\hline
 & \textbf{1960} & \textbf{1965} & \textbf{1970} & \textbf{1975}\\ 
%\textsc{Nodes} & 146\,253 & 174\,826 & 210\,527 & 257\,896\\ 
%\textsc{Edges} & 11\,197\,509 & 12\,649\,114 & 14\,209\,908 & 16\,080\,065\\ 
%\textsc{Impr ($k=1$)} & 122.63 & 162.06 & 211.05 & 285.57\\ 
%\textsc{Impr ($k=10$)} & 80.50 & 111.51 & 159.32 & 221.07\\ 
\hline
\textsc{1} &               Flowers, Bess  &               Flowers, Bess  &               Flowers, Bess  &               Flowers, Bess \\ 
\textsc{2} &            Harris, Sam (II)  &            Harris, Sam (II)  &            Harris, Sam (II)  &            Harris, Sam (II) \\ 
\textsc{3} &            Farnum, Franklyn  &            Farnum, Franklyn  &              Tamiroff, Akim  &              Tamiroff, Akim \\ 
\textsc{4} &          Miller, Harold (I)  &          Miller, Harold (I)  &            Farnum, Franklyn  &               Welles, Orson \\ 
\textsc{5} &                 Chefe, Jack  &              Holmes, Stuart  &          Miller, Harold (I)  &              Sayre, Jeffrey \\ 
\textsc{6} &              Holmes, Stuart  &              Sayre, Jeffrey  &              Sayre, Jeffrey  &          Miller, Harold (I) \\ 
\textsc{7} &               Steers, Larry  &                 Chefe, Jack  &          Quinn, Anthony (I)  &            Farnum, Franklyn \\ 
\textsc{8} &               Par\`is, Manuel  &               Par\`is, Manuel  &         O'Brien, William H.  &             Kemp, Kenner G. \\ 
\textsc{9} &         O'Brien, William H.  &         O'Brien, William H.  &              Holmes, Stuart  &          Quinn, Anthony (I) \\ 
\textsc{10} &               Sayre, Jeffrey  &            Stevens, Bert (I)  &            Stevens, Bert (I)  &          O'Brien, William H. \\ 
\hline
%\multicolumn{5}{c}{} \\ 
\hline
 & \textbf{1980} & \textbf{1985} & \textbf{1990} & \textbf{1995}\\ 
%\textsc{Nodes} & 310\,278 & 375\,322 & 463\,078 & 557\,373\\ 
%\textsc{Edges} & 18\,252\,462 & 20\,970\,510 & 24\,573\,288 & 28\,542\,684\\ 
%\textsc{Impr ($k=1$)} & 380.52 & 513.40 & 719.21 & 971.11 \\ 
%\textsc{Impr ($k=10$)} & 296.24 & 416.27 & 546.77 & 694.72 \\ 
\hline
\textsc{1} &               Flowers, Bess  &               Welles, Orson  &               Welles, Orson  &        Lee, Christopher (I) \\ 
\textsc{2} &            Harris, Sam (II)  &               Flowers, Bess  &             Carradine, John  &               Welles, Orson \\ 
\textsc{3} &               Welles, Orson  &            Harris, Sam (II)  &               Flowers, Bess  &          Quinn, Anthony (I) \\ 
\textsc{4} &              Sayre, Jeffrey  &          Quinn, Anthony (I)  &        Lee, Christopher (I)  &           Pleasence, Donald \\ 
\textsc{5} &          Quinn, Anthony (I)  &              Sayre, Jeffrey  &            Harris, Sam (II)  &               Hitler, Adolf \\ 
\textsc{6} &              Tamiroff, Akim  &             Carradine, John  &          Quinn, Anthony (I)  &             Carradine, John \\ 
\textsc{7} &          Miller, Harold (I)  &             Kemp, Kenner G.  &           Pleasence, Donald  &               Flowers, Bess \\ 
\textsc{8} &             Kemp, Kenner G.  &          Miller, Harold (I)  &              Sayre, Jeffrey  &             Mitchum, Robert \\ 
\textsc{9} &            Farnum, Franklyn  &            Niven, David (I)  &               Tovey, Arthur  &            Harris, Sam (II) \\ 
\textsc{10} &             Niven, David (I)  &               Tamiroff, Akim  &                Hitler, Adolf  &               Sayre, Jeffrey \\ 
\hline
%\multicolumn{5}{c}{} \\ 
\hline
 & \textbf{2000} & \textbf{2005} & \textbf{2010} & \textbf{2014}\\ 
%\textsc{Nodes} & 681\,358 & 880\,032 & 1\,237\,879 & 1\,797\,446\\ 
%\textsc{Edges} & 33\,564\,142 & 41\,079\,259 & 53\,625\,608 & 72\,880\,156\\ 
%\textsc{Impr ($k=1$)} & 1326.53 & 1897.31 & 2869.14 & 2601.52\\
%\textsc{Impr ($k=10$)} & 838.53 & 991.89 & 976.63 & 1390.32\\
 \hline
\textsc{1} &        Lee, Christopher (I)  &               Hitler, Adolf  &               Hitler, Adolf  &         Madsen, Michael (I) \\ 
\textsc{2} &               Hitler, Adolf  &        Lee, Christopher (I)  &        Lee, Christopher (I)  &                Trejo, Danny \\ 
\textsc{3} &           Pleasence, Donald  &                Steiger, Rod  &              Hopper, Dennis  &               Hitler, Adolf \\ 
\textsc{4} &               Welles, Orson  &      Sutherland, Donald (I)  &          Keitel, Harvey (I)  &           Roberts, Eric (I) \\ 
\textsc{5} &          Quinn, Anthony (I)  &           Pleasence, Donald  &            Carradine, David  &             De Niro, Robert \\ 
\textsc{6} &                Steiger, Rod  &              Hopper, Dennis  &      Sutherland, Donald (I)  &               Dafoe, Willem \\ 
\textsc{7} &             Carradine, John  &          Keitel, Harvey (I)  &               Dafoe, Willem  &          Jackson, Samuel L. \\ 
\textsc{8} &      Sutherland, Donald (I)  &          von Sydow, Max (I)  &          Caine, Michael (I)  &          Keitel, Harvey (I) \\ 
\textsc{9} &             Mitchum, Robert  &          Caine, Michael (I)  &               Sheen, Martin  &            Carradine, David \\ 
\textsc{10} &                Connery, Sean  &                Sheen, Martin  &                    Kier, Udo  &         Lee, Christopher (I) \\
\hline
\end{tabular}}
\end{scriptsize}
\end{table*}
\begin{table*}[!h]
\begin{scriptsize}
\centering
\tbl{Detailed improvement factors on the IMDB actor graph. \label{tab:imdb2}}{
\begin{tabular}{|l|c|c|c|c|}
\hline
\textsc{Year} & \textbf{1940} & \textbf{1945} & \textbf{1950} & \textbf{1955}\\ 
\hline
\textsc{Nodes} & 69\,011 & 83\,068 & 97\,824 & 120\,430\\ 
\textsc{Edges} & 3\,417\,144 & 5\,160\,584 & 6\,793\,184 & 8\,674\,159\\ 
\textsc{Impr ($k=1$)} & 51.74 & 61.46 & 67.50 & 91.46 \\ 
\textsc{Impr ($k=10$)} & 32.95 & 40.73 & 44.72 & 61.52 \\ 
\hline
%\multicolumn{5}{c}{\vspace{-.5\baselineskip}} \\ 
\hline
\textsc{Year} & \textbf{1960} & \textbf{1965} & \textbf{1970} & \textbf{1975}\\ 
\hline
\textsc{Nodes} & 146\,253 & 174\,826 & 210\,527 & 257\,896\\ 
\textsc{Edges} & 11\,197\,509 & 12\,649\,114 & 14\,209\,908 & 16\,080\,065\\ 
\textsc{Impr ($k=1$)} & 122.63 & 162.06 & 211.05 & 285.57\\ 
\textsc{Impr ($k=10$)} & 80.50 & 111.51 & 159.32 & 221.07\\ 
\hline
%\multicolumn{5}{c}{} \\ 
\hline
\textsc{Year} & \textbf{1980} & \textbf{1985} & \textbf{1990} & \textbf{1995}\\
\hline
\textsc{Nodes} & 310\,278 & 375\,322 & 463\,078 & 557\,373\\ 
\textsc{Edges} & 18\,252\,462 & 20\,970\,510 & 24\,573\,288 & 28\,542\,684\\ 
\textsc{Impr ($k=1$)} & 380.52 & 513.40 & 719.21 & 971.11 \\ 
\textsc{Impr ($k=10$)} & 296.24 & 416.27 & 546.77 & 694.72 \\ 
\hline
%\multicolumn{5}{c}{} \\ 
\hline
\textsc{Year} & \textbf{2000} & \textbf{2005} & \textbf{2010} & \textbf{2014}\\ 
\hline
\textsc{Nodes} & 681\,358 & 880\,032 & 1\,237\,879 & 1\,797\,446\\ 
\textsc{Edges} & 33\,564\,142 & 41\,079\,259 & 53\,625\,608 & 72\,880\,156\\ 
\textsc{Impr ($k=1$)} & 1326.53 & 1897.31 & 2869.14 & 2601.52\\
\textsc{Impr ($k=10$)} & 838.53 & 991.89 & 976.63 & 1390.32\\
\hline
\end{tabular}}
\end{scriptsize}
\end{table*}
\end{document}